\documentclass[a4paper,UKenglish,cleveref]{lipics-v2021}
\usepackage{amsmath,amssymb,amsthm}
\usepackage{mathtools}
\usepackage{xspace}
\usepackage{tikz}
\usetikzlibrary{automata,positioning,shapes,fit}
\usepackage{thmtools}
\usepackage{enumitem}
\usepackage{multirow}
\usepackage{nicefrac}
\usepackage[basic,bold]{complexity}
\usepackage{todonotes}
\usepackage[capitalise]{cleveref}
\usepackage{float}

\makeatletter
\g@addto@macro\bfseries{\boldmath}
\makeatother

\newcommand{\Au}{\mathcal{A}}
\newcommand{\B}{\mathcal{B}}
\newcommand{\T}{\mathcal{T}}

\newcommand{\N}{\mathbb{N}}
\newcommand{\Z}{\mathbb{Z}}
\newcommand{\true}{\mathrm{true}}
\newcommand{\false}{\mathrm{false}}
\newcommand{\divs}{\mathrel{|}}
\newcommand{\lcm}{\mathrm{lcm}}

\newcommand{\defeq}{\overset{\text{def}}{=}}
\renewcommand{\vec}[1]{\boldsymbol{#1}}

\newcommand{\apad}{$\forall$PAD\xspace}
\newcommand{\epad}{$\exists$PAD\xspace}

\newclass{\NTWOEXP}{2NEXP}
\newclass{\coNTWOEXP}{co2NEXP}
\newclass{\NTHREEEXP}{3NEXP}
\newclass{\coNTHREEEXP}{co3NEXP}

\newcommand{\aerpad}{$\forall\exists_R$PAD\xspace}
\newcommand{\earpad}{$\exists\forall_R$PAD\xspace}
\newcommand{\aerpadplus}{$\forall\exists_R$PAD$^+$\xspace}

\newcommand{\reach}{\varphi_{\mathrm{reach}}}
\newcommand{\reachnt}{\varphi_{\mathrm{reach-nt}}}

\newcommand{\cob}{\varphi_{\textnormal{B\"uchi}}}

\newcommand{\Op}{\mathit{Op}}
\newcommand{\cu}{\mathit{CU}}
\newcommand{\pu}{\mathit{PU}}
\newcommand{\zt}{\mathit{ZT}}
\newcommand{\pt}{\mathit{PT}}
\newcommand{\op}{\mathit{op}}
\newcommand{\subata}[1]{\T^{\mathrm{#1}}_{\mathrm{sub}}}

\newcommand{\first}{\mathrm{first}}

\newcommand{\weight}{\mathrm{weight}}

\newcommand{\lef}[1]{\mathrm{left}{#1}}
\newcommand{\righ}[1]{\mathrm{right}{#1}}
\newcommand{\final}{\mathrm{final}}
\newcommand{\simu}{\displaystyle{simulation}}
\newcommand{\val}{\displaystyle{validation}}
\newcommand{\vio}{\displaystyle{violation}}

\usepackage{xcolor}
\definecolor{codegreen}{rgb}{0,0.6,0}
\definecolor{codegray}{rgb}{0.5,0.5,0.5}
\definecolor{codepurple}{rgb}{0.58,0,0.82}
\definecolor{backcolour}{rgb}{0.95,0.95,0.92}
\lstdefinestyle{mystyle}{
  commentstyle=\color{codegreen},
  keywordstyle=\color{magenta},
  numberstyle=\tiny\color{codegray},
  stringstyle=\color{codepurple},
  basicstyle=\ttfamily\footnotesize,
  breakatwhitespace=false,
  breaklines=true,
  keepspaces=true,
  numbers=left,
  numbersep=5pt,
  showspaces=false,
  showstringspaces=false,
  showtabs=false,
  tabsize=2
}
\title{Revisiting Parameter Synthesis for One-Counter Automata}

\author{Guillermo A. P\'erez}{University of Antwerp -- Flanders Make, Belgium}{guillermoalberto.perez@uantwerpen.be}{https://orcid.org/0000-0002-1200-4952}{}

\author{Ritam Raha}{University of Antwerp, Belgium\\
 LaBRI, University of Bordeaux, France}{ritam.raha@uantwerpen.be}{https://orcid.org/0000-0003-1467-1182
}{}

\authorrunning{G. A. P\'erez and R. Raha}

\Copyright{Guillermo A. P\'erez and Ritam Raha} 



\acknowledgements{
We thank all anonymous reviewers who carefully read earlier version of this work and helped us polish the paper. We also thank Micha\"el Cadilhac, Nathana\"el Fijalkow, Philip Offtermatt, and Mikhail R. Starchak for useful feedback; Antonia Lechner and James Worrell for having brought the inconsistencies in~\cite{bi05,lechner15} to our attention; Radu Iosif and Marius Bozga for having pointed out precisely what part of the argument in both papers is incorrect.
}

\nolinenumbers 

\hideLIPIcs  

\EventEditors{John Q. Open and Joan R. Access}
\EventNoEds{2}
\EventLongTitle{42nd Conference on Very Important Topics (CVIT 2016)}
\EventShortTitle{CVIT 2016}
\EventAcronym{CVIT}
\EventYear{2016}
\EventDate{December 24--27, 2016}
\EventLocation{Little Whinging, United Kingdom}
\EventLogo{}
\SeriesVolume{42}
\ArticleNo{23}

\begin{document}

\maketitle

\begin{abstract}
  We study the synthesis problem for one-counter automata with
  parameters. One-counter automata are obtained by extending classical
  finite-state automata with a counter whose value can range over
  non-negative integers and be tested for zero. The updates and tests
  applicable to the counter can further be made parametric by introducing a
  set of integer-valued variables called parameters. The synthesis problem
  for such automata asks whether there exists a valuation of the parameters
  such that all infinite runs of the automaton satisfy some $\omega$-regular
  property. Lechner showed that (the complement of) the problem can be encoded
  in a restricted one-alternation fragment of Presburger arithmetic with
  divisibility. In this work (i) we argue that said fragment, called
  \aerpadplus, is unfortunately undecidable. Nevertheless, by a careful
  re-encoding of the problem into a decidable restriction of \aerpadplus, (ii)
  we prove that the synthesis problem is decidable in general and in
  \NTWOEXP{} for several fixed $\omega$-regular properties. Finally, (iii) we
  give polynomial-space algorithms for the special cases of the problem where
  parameters can only be used in counter tests.
\end{abstract}

\section{Introduction}
\label{intro}

Our interest in one-counter automata (OCA)
with parameters stems from their usefulness
as models of the behaviour of programs whose control flow is determined by \emph{counter variables}.

\begin{lstlisting}[language=Python]
def funprint(x):
    i = 0
    i += x
    while i >= 0:
        if i == 0:
            print("Hello")
        if i == 1:
            print("world")
        if i >= 2:
            assert(False)
        i -= 1
    # end program
\end{lstlisting}

\noindent
Indeed, the executions of such a program can be over-approximated by its
control-flow graph (CFG)~\cite{allen70}. The CFG can be leveraged to get a
\emph{conservative response} to interesting questions about the program, such as: ``is there a value of $x$ such that the false assertion is
avoided?'' The CFG abstracts away all variables and their values (see
Figure~\ref{fig:cfg-oca}) and this introduces non-determinism. Hence, the
question becomes: ``is it the case that all paths from the initial vertex
avoid the one labelled with $10$?'' In this particular example, the abstraction
is too \emph{coarse} and thus we obtain a false negative.  In such cases, the
abstraction of the program should be refined~\cite{cousot10}. A natural
refinement of the CFG in this context is obtained by tracking the value of
$i$ (cf. program graphs in~\cite{bk08}). The result is an OCA with parameters
such that:
For $x \in \{0,1\}$ it
has no run
that reaches the state labelled with $10$. This is an instance of a
\emph{safety (parameter) synthesis problem} for which the answer is positive.

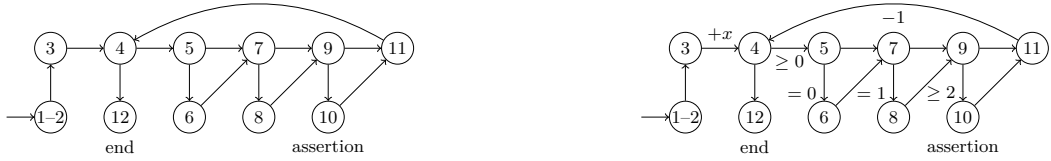
\begin{figure}
  \centering
\scalebox{0.8}{
  \begin{tikzpicture}[node distance=0.6cm,every state/.style={inner
    sep=0,minimum size=15pt},every node/.style={font=\footnotesize}]
    \node[state,initial,initial text={}](q1){1--2};
    \node[state,above= of q1](q3){3};
    \node[state,right= of q3](q4){4};
    \node[state,right= of q4](q5){5};
    \node[state,below= of q5](q6){6};
    \node[state,right= of q5](q7){7};
    \node[state,below= of q7](q8){8};
    \node[state,right= of q7](q9){9};
    \node[state,right= of q9](q11){11};
    \node[state,below= of q9,label=below:{assertion}](q10){10};
    \node[state,below= of q4,label=below:{end}](q12){12};
    \path[->]
      (q1) edge (q3)
      (q3) edge (q4)
      (q4) edge (q5)
      (q5) edge (q6)
      (q5) edge (q7)
      (q6) edge (q7)
      (q7) edge (q8)
      (q7) edge (q9)
      (q8) edge (q9)
      (q9) edge (q10)
      (q9) edge (q11)
      (q4) edge (q12)
      (q10) edge (q11)
      (q11) edge[bend right] (q4)
      ;
  \end{tikzpicture}}
  \hfill
 \scalebox{0.8}{
  \begin{tikzpicture}[node distance=0.6cm,every state/.style={inner
    sep=0,minimum size=15pt},every node/.style={font=\footnotesize}]
    \node[state,initial,initial text={}](q1){1--2};
    \node[state,above= of q1](q3){3};
    \node[state,right= of q3](q4){4};
    \node[state,right= of q4](q5){5};
    \node[state,below= of q5](q6){6};
    \node[state,right= of q5](q7){7};
    \node[state,below= of q7](q8){8};
    \node[state,right= of q7](q9){9};
    \node[state,right= of q9](q11){11};
    \node[state,below= of q9,label=below:{assertion}](q10){10};
    \node[state,below= of q4,label=below:{end}](q12){12};
    \path[->,auto]
      (q1) edge (q3)
      (q3) edge node{$+x$} (q4)
      (q4) edge node[swap]{$\geq 0$} (q5)
      (q5) edge node[swap,pos=0.8]{$=0$} (q6)
      (q5) edge (q7)
      (q6) edge (q7)
      (q7) edge node[swap,pos=0.8]{$=1$}(q8)
      (q7) edge (q9)
      (q8) edge (q9)
      (q9) edge node[swap,pos=0.8]{$\geq 2$}(q10)
      (q9) edge (q11)
      (q4) edge (q12)
      (q10) edge (q11)
      (q11) edge[bend right] node{$-1$} (q4)
      ;
  \end{tikzpicture}}
  \caption{On the left, the CFG with vertex labels corresponding to source
  code line numbers; on the right, the CFG extended by tracking the value of
  $i$}
  \label{fig:cfg-oca}
\end{figure}

In this work, we focus on the parameter synthesis problems for given OCA with
parameters and do not consider the problem of obtaining such an OCA from a
program (cf.~\cite{fkp14}).

Counter automata~\cite{Minsky10} are a classical model that extend
finite-state automata with integer-valued counters. These have been shown to
be useful in modelling complex systems, such as programs with lists and XML
query evaluation algorithms~\cite{BouajjaniBHIMV06,ChiticR04}. Despite their usefulness as a modelling formalism, it is known that two counters suffice for counter automata to become Turing powerful. In particular, this means that most interesting questions about them are undecidable~\cite{Minsky10}. To circumvent this, several restrictions of the model have been studied in the literature, e.g. reversal-bounded counter automata~\cite{IbarraSDBK02}
and automata with a single counter. In this work we focus on an extension of the latter: OCA with parametric updates and parametric tests.

An existential version of the synthesis problems for OCA with parameters was
considered by G\"oller et al.~\cite{ghow10} and Bollig et al.~\cite{bqs19}.
They ask whether there exist a valuation of the parameters and a run of the
automaton which satisfies a given $\omega$-regular property. This is in contrast
to the present problem where we quantify runs universally. (This is
required for the conservative-approximation use case described in the example
above.) We note that, of those two works, only~\cite{ghow10} considers OCA with
parameters allowed in both counter updates and counter tests
while~\cite{bqs19} studies OCA with parametric tests only. In this paper,
unless explicitly stated otherwise, we focus on OCA with parametric tests and
updates like in~\cite{ghow10}.  Further note that the model we study has an
asymmetric set of tests that can be applied to the counter: lower-bound tests,
and equality tests (both parametric and non-parametric). The primary reason for this is
that adding upper-bound tests results in a model for which even the
decidability of the (arguably simpler) existential reachability synthesis
problem is a long-standing open problem~\cite{BundalaO17}. Namely, the
resulting model corresponds to Ibarra's \emph{simple
programs}~\cite{ijtw95}.

In both~\cite{ghow10} and~\cite{bqs19}, the synthesis problems for OCA with
parameters were stated as open. Later, Lechner~\cite{lechner15} gave an
encoding for the complement of the synthesis problems into a one-alternation
fragment of Presburger arithmetic with divisibility (PAD). Her encoding
relies on work by Haase et al.~\cite{hkow09}, which shows how to
compute a linear-arithmetic representation of the reachability relation of OCA
(see~\cite{licwx20} for an implementation). In the same work, Haase et al.
show that the same can be achieved for OCA with parameters using the
divisibility predicate. In~\cite{lechner15}, Lechner goes on to consider the
complexity of (validity of sentences in) the language corresponding to the
one-alternation fragment her encoding targets. An earlier paper~\cite{bi05} by
Bozga and Iosif argues that the fragment is decidable and Lechner carefully
repeats their argument while leveraging bounds on the bitsize of solutions of
existential PAD formulas~\cite{LechnerOW15} to argue the complexity of the
fragment is \coNTWOEXP{}. For $\omega$-regular properties given as a linear
temporal logic (LTL) formula, her encoding is exponential in the formula and
thus it follows that the LTL synthesis problem is decidable and in
\NTHREEEXP{}.

\subparagraph{Problems in the literature.}
\emph{Presburger arithmetic} is the first-order theory of $\langle
\Z,0,1,+,< \rangle$.
\emph{Presburger arithmetic with divisibility (PAD)} is the
extension of PA obtained when we add a binary \emph{divisibility predicate}.
The resulting language 
is undecidable~\cite{robinson49}. In fact, a single quantifier
alternation already allows to encode general multiplication, thus becoming
undecidable~\cite{lipshitz81}. However, the purely existential ($\Sigma_0$)
and purely universal ($\Pi_0$) fragments have been shown to be
decidable~\cite{Beltyukov1980,lipshitz78}.

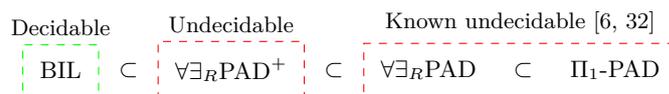
\begin{figure}
  \centering
  \small
  \begin{tikzpicture}
    \node(bil) {BIL};
    \node[right= of bil](aerpadplus) {\aerpadplus};
    \node[right= of aerpadplus](aerpad) {\aerpad};
    \node[right= of aerpad](sigma1) {$\Pi_1$-PAD};
    \node[fit=(aerpad)(sigma1),label=above:{\footnotesize Known
      undecidable~\cite{bi05,robinson49}},dashed,draw=red] (undec) {};
    \node[fit=(aerpadplus),label=above:{\footnotesize Undecidable},dashed,draw=red] (undec2)
      {};
    \node[fit=(bil),label=above:{\footnotesize Decidable},dashed,draw=green] (dec) {};
    \path
      (bil) edge[draw=none] node{$\subset$} (aerpadplus)
      (aerpadplus) edge[draw=none] node{$\subset$} (aerpad)
      (aerpad) edge[draw=none] node{$\subset$} (sigma1)
      ;
  \end{tikzpicture}
  \caption{Syntactical fragments of PAD
  ordered w.r.t. their language (of sentences)}
\end{figure}

The target of Lechner's encoding is \aerpadplus, a subset of all sentences in
the $\Pi_1$-fragment of PAD. Such sentences look as follows:
\(
  \forall \vec{x} \exists \vec{y}
  \bigvee_{i \in I} \bigwedge_{j \in J_i}
  f_j(\vec{x}) \divs g_j(\vec{x},\vec{y}) \land
  \varphi_i(\vec{x},\vec{y})
\)
where $\varphi$ is a quantifier-free PAD formula without divisibility. Note
that all divisibility constraints appear in positive form (hence the $^+$) and
that, within divisibility constraints, the existentially-quantified variables
$y_i$ appear only on the right-hand side (hence the $\exists_\text{R}$).
In~\cite{bi05}, the authors give a quantifier-elimination procedure for
sentences in a further restricted fragment we call the Bozga-Iosif-Lechner
fragment (BIL) that is based on ``symbolically applying'' the generalized
Chinese remainder theorem (CRT)~\cite{sm58}. Their procedure does not 
eliminate all quantifiers but rather yields a sentence in the $\Pi_0$-fragment
of PAD. (Decidability of the BIL language would then follow from the result of
Lipshitz~\cite{lipshitz78}.) Then, they \emph{briefly argue} how the algorithm generalizes to
\aerpadplus. There are two crucial problems in the argument from~\cite{bi05}
that we have summarized here (and which were reproduced in Lechner's work):
First, the quantifier-elimination procedure of Bozga and Iosif does
    not directly work for BIL. Indeed, not all BIL sentences satisfy the
    conditions required for the CRT to be applicable as used in their
    algorithm.
    Second, there is no way to generalize their algorithm to \aerpadplus
    since the language is undecidable. Interestingly, undecidability follows
    directly from other results in~\cite{bi05,lechner15}.   
In Lechner's thesis~\cite{lechner16}, the result from~\cite{bi05}
was stated as being under review. Correspondingly, the decidability of the
synthesis problems for OCA with parameters was only stated conditionally on
\aerpadplus being decidable.

\begin{table}[t]    
\centering
    \small
    \begin{tabular}{|l | c |c|}
    \hline
         & Lower bound & Upper bound \\
         \hline
          LTL& \PSPACE-hard~\cite{sc85} & in \NTHREEEXP{} (Cor.~\ref{cor:ltl-synth})\\
          \hline
         Reachability& \coNP-hard (Prop.~\ref{prop_synthreach_nopara}) &
         \multirow{2}{*}{in \NTWOEXP{} (Thm.~\ref{thm-synthreach})}\\ 
         \cline{1-2}
         Safety, B\"uchi, coB\"uchi & $\NP^\NP$-hard~\cite{lechner15,lechner16} & \\
         \hline
    \end{tabular}
    \caption{Known and new complexity bounds for parameter synthesis problems}

\label{tab:complexity}
\end{table}

\subparagraph*{Our contribution.}
In~\Cref{sec:PADBIL}, using developments from~\cite{bi05,lechner15}, we argue
that \aerpadplus is undecidable (\Cref{thm:undec-aerpadplus}). Then, in the same
section, we ``fix'' the definition of the BIL fragment by adding to it a necessary
constraint so that the quantifier-elimination procedure from~\cite{bi05} works
correctly. For completeness, and to clarify earlier mistakes in the
literature, we recall Lechner's analysis of the algorithm and conclude, just as
she did, that the complexity of BIL is in \coNTWOEXP{}~\cite{lechner15}
(\Cref{thm:bildec}). After some preliminaries regarding OCA with parameters
in~\Cref{sec:prelims}, we re-establish decidability of various synthesis
problems in~\Cref{sec:synthOCA} (\Cref{thm-synthreach} and
\Cref{cor:ltl-synth}, see ~\Cref{tab:complexity} for a summary). To do so, we
follow Lechner's original idea from~\cite{lechner15} to encode them into
\aerpadplus sentences. However, to ensure we obtain a BIL sentence, several parts
of her encoding have to be adapted. Finally, in~\Cref{sec:ocapt} we make small
modifications to the work of Bollig et al.~\cite{bqs19} to give more efficient
algorithms that are applicable when only tests have parameters
(\Cref{thm_synthreach_ocapt} and \Cref{cor:ltl-synth2}).


\section{Presburger Arithmetic with divisibility}
\label{sec:PADBIL}

\emph{Presburger arithmetic (PA)} is the first-order theory over $\langle
\Z,0,1,+,< \rangle$ where $+$ and $<$ are the standard addition and ordering
of integers. \emph{Presburger arithmetic with divisibility (PAD)} is the
extension of PA obtained when we add the binary divisibility predicate
$\divs$, where for all $a,b \in \Z$ we have
\(
    a \divs b \iff \exists c \in \Z: b = ac.
\)
Let $X$ be a finite set of first-order variables. A \emph{linear polynomial}
over $\vec{x} = (x_1, \dots ,x_n) \in X^n$ is given by the syntax rule:
\(
  p(\vec{x}) ::= \sum_{1 \leq i \leq n} a_ix_i+b,
\)
where the $a_i$, $b$ and the first-order variables from
$\vec{x}$ range over $\Z$.  In general, quantifier-free PAD formulas have the
grammar:
\(
    \varphi ::= {} \varphi_1 \land \varphi_2 \:|\: \lnot \varphi \:|\:
       f(\vec{x}) \mathrel{P} g(\vec{x}),
\)
where $P$ can be the order predicate  $<$ or the divisibility predicate
$\divs$, and $f,g$ are linear polynomials. We define the standard Boolean
abbreviation $\varphi_1 \lor \varphi_2 \iff \lnot(\lnot \varphi_1 \land \lnot
\varphi_2)$. Moreover we introduce the abbreviations $f(x) \leq g(x) \iff
f(x) < g(x) + 1$ and $f(x)=g(x) \iff f(x) \leq g(x) \land g(x) \leq f(x)$.

The \emph{size} $|\varphi|$ of a PAD formula
$\varphi$ is defined by structural induction over $|\varphi|$: For a linear
polynomial $p(\vec{x})$ we define $|p(\vec{x})|$ as the number of symbols
required to write it if the coefficients are given in binary. Then, we define
$|\varphi_1 \land \varphi_2| \defeq |\varphi_1| + |\varphi_2| + 1$, $|\lnot
\varphi| \defeq |\exists x. \varphi| \defeq  |\varphi| + 1$, $|f(\vec{x})
\mathrel{P} g(\vec{x})| \defeq |f(\vec{x})| + |g(\vec{x})| + 1$.

\subsection{Allowing one restricted alternation}
We define the language \aerpad of all PAD sentences allowing a universal
quantification over some variables, followed by an existential quantification
over variables that may not appear on the left-hand side of divisibility
constraints. Formally, \aerpad is the set of all PAD sentences of the form:
\(
  \forall x_1 \dots \forall x_n \exists y_1 \dots \exists y_m
  \varphi(\vec{x},\vec{y})
\)
where $\varphi$ is a quantifier-free PAD formula and all its divisibility constraints
are of the form
\(
  f(\vec{x}) \divs g(\vec{x},\vec{y}).
\)

\subparagraph{Positive-divisibility fragment.} We denote by \aerpadplus the
subset of \aerpad sentences $\varphi$ where the negation operator can only be
applied to the order predicate $<$ and the only other Boolean operators
allowed are conjunction and disjunction. In other words, \aerpadplus is a
restricted negation normal form in which divisibility predicates cannot be
negated.
Lechner showed in~\cite{lechner15} that all
\aerpad sentences can be translated into \aerpadplus sentences.
\begin{proposition}[Lechner's trick~\cite{lechner15}]\label{pro:lechner}
    For all $\varphi_1$ in \aerpad one can compute $\varphi_2$ in \aerpadplus such that
    $\varphi_1$ is true if and only if $\varphi_2$ is true.
\end{proposition}

\subsection{Undecidability of both one-alternation fragments}
We will now prove that the language \aerpadplus is undecidable, that is, to determine
whether a given sentence from \aerpadplus is true is an undecidable problem.
\begin{theorem}\label{thm:undec-aerpadplus}
    The language \aerpadplus is undecidable.
\end{theorem}
From~\Cref{pro:lechner} it follows that arguing \aerpad is undecidable
suffices to prove the theorem. The latter was proven in~\cite{bi05}.
More precisely, they show the complementary language is
undecidable. Their argument consists in defining the least-common-multiple
predicate, the squaring predicate, and subsequently integer multiplication.
Undecidability thus follows from the MRDP theorem~\cite{matijasevic70} which
states that satisfiability for such equations (i.e. Hilbert's 10th problem) is
undecidable. Hence, \Cref{thm:undec-aerpadplus} is a direct consequence of the following result.

\begin{proposition}[From~\cite{bi05}]
    \label{pro:undec-earpad}
    The language \aerpad is undecidable.
\end{proposition}


\subsection{The Bozga-Iosif-Lechner fragment}

The Bozga-Iosif-Lechner (BIL) fragment is the
set of all \aerpadplus sentences of the form:
\[
  \forall x_1 \dots \forall x_n \exists y_1\dots
  \exists y_m
  (\vec{x} < 0) \lor
   \bigvee_{i \in I}
  \bigwedge_{j \in J_i} \left(f_j(\vec{x}) \divs g_j(\vec{x},\vec{y}) \land
  f_j(\vec{x}) > 0\right)
  \land \varphi_i(\vec{x})
  \land \vec{y} \geq \vec{0}
\] 
where $I,J_i \subseteq \mathbb{N}$ are all finite index sets, the $f_j$ and
$g_j$ are linear polynomials and the $\varphi_i(\vec{x})$ are quantifier-free
PA formulas over the variables $\vec{x}$. Note that, compared to
\aerpadplus, BIL sentences only constraint non-negative values of $\vec{x}$. (This technicality is necessary due to our second constraint below.) For readability, henceforth, we omit $(\vec{x} < 0)$ and just assume the $\vec{x}$ take non-negative integer values, i.e. from $\mathbb{N}$. Additionally, it introduces the following three important constraints:
\begin{enumerate}
  \item The $\vec{y}$ variables may only
    appear on the right-hand side of divisibility constraints.
  \item All divisibility constraints $f_j(\vec{x}) \divs g_j(\vec{x},\vec{y})$
    are conjoined with $f_j(\vec{x}) > 0$.
  \item The $\vec{y}$ variables are only allowed to take non-negative values.
\end{enumerate}
It should be clear that the first constraint is necessary to avoid
undecidability. Indeed, if the $\vec{y}$ variables were allowed in the PA
formulas $\varphi_i(\vec{x})$ then we could circumvent the restrictions of
where they appear in divisibilities by using equality constraints. The second
constraint is similar in spirit. Note that if $a = 0$ then $a \divs b$ holds
if and only if $b = 0$ so if the left-hand side of divisibility constraints is
allowed to be $0$ then we can encode PA formulas on $\vec{x}$ and $\vec{y}$ as
before. Also, the latter (which was missing
in~\cite{bi05,lechner15}) will streamline the application of the generalized
Chinese remainder theorem in the algorithm described in the sequel. While the third constraint is not required for decidability, it is convenient to include it for \Cref{sec:encoding}, where we encode instances of the synthesis problem into the BIL fragment. 

In the rest of this section, we recall the decidability proof by Bozga and
Iosif~\cite{bi05} and refine Lechner's analysis~\cite{lechner15} to obtain the
following complexity bound.

\begin{theorem}\label{thm:bildec}
  The BIL-fragment language is decidable
  in \coNTWOEXP.
\end{theorem}

The idea of the proof is as follows: We start from a BIL sentence. First, we
use the \emph{generalized Chinese remainder theorem} (CRT, for short) to
replace all of the existentially quantified variables in it with a single
universally quantified variable. We thus obtain a sentence in \apad (i.e. the $\Pi_0$-fragment of PAD) and argue
that the desired result follows from the bounds on the bitsize of
satisfying assignments for existential PAD formulas
\cite{LechnerOW15}.

\begin{theorem}[Generalized Chinese remainder theorem~\cite{sm58}]
  Let $m_i \in \N_{>0}$, $a_i,r_i \in \mathbb{Z}$ for $1 \leq i
  \leq n$. Then, there exists $x \in \mathbb{Z}$ such that
  \( \bigwedge_{i=1}^n m_i \divs (a_ix - r_i) \)
  if and only if:
  \[
    \bigwedge_{1 \leq i, j \leq n}
    \gcd(a_im_j, a_jm_i) \divs (a_ir_j - a_jr_i)
    \land
    \bigwedge_{i=1}^n \gcd(a_i,m_i) \divs r_i.
  \]
  The solution for $x$ is unique modulo $\lcm(m'_1,\dots,m'_n)$, where $m'_i =
  \nicefrac{m_i}{\gcd(a_i,m_i)}$.
\end{theorem}

From a BIL sentence, we apply the CRT to the rightmost
existentially quantified variable and get a sentence with one less
existentially quantified variable and with $\gcd$-expressions. Observe that
the second restriction we highlighted for the BIL fragment (the conjunction
with $f_j(\vec{x}) > 0$) is necessary for the correct application of the
CRT.  We will later argue that we can remove the $\gcd$
expressions to obtain a sentence in \apad.

\begin{example}
  Consider the sentence:
  \[
      \forall x \exists y_1 \exists y_2 \bigvee_{i \in I} \bigwedge_{j \in
      J_i} \left(f_j(x) \divs g_j(x,\vec{y}) \land f_j(x) > 0 \right) \land
      \varphi_i(x) \land  \vec{y} \geq \vec{0}.
  \]
  Let $\alpha_j$ denote the coefficient of $y_2$ in $g_j(x,\vec{y})$ and
  $r_j(x,y_1) \defeq -(g_j(x,\vec{y}) - \alpha_j y_2)$.  We can rewrite the
  above sentence as $\forall x \exists y_1 \bigvee_{i\in I} \psi_i(x,y_1)
  \land \varphi'_i(x) \land y_1 \geq 0$ where:
  \begin{align*}
    \psi_i(x,y_1) ={} & \exists
  y_2 \bigwedge_{j \in J_i} (f_j(x) \divs (\alpha_j y_2 - r_j(x,y_1))) \land
  y_2 \geq 0, \text{ and}\\
  \varphi'_i(x) ={} & \varphi_i(x) \land \bigwedge_{j \in J_i}
  f_j(x) > 0.
  \end{align*}
  Applying the CRT, $\psi_i(x,y_1)$ can equivalently
  be written as follows:
  \[
     \bigwedge_{j,k \in J_i}
     \gcd(\alpha_k f_j(x), \alpha_j f_k(x)) \divs (\alpha_j
     r_k(x,y_1) - \alpha_k r_j(x,y_1))
     \land \bigwedge_{j \in J_i} \gcd(\alpha_j, f_j(x)) \divs
     r_j(x,y_1).
  \]
  Note that we have dropped the $y_2 \geq 0$ constraint without loss of
  generality since the CRT states that the set of solutions forms
  an arithmetic progression containing infinitely many positive (and negative)
  integers.  This means the constraint will be trivially satisfied for any
  valuation of $x$ and $y_1$ which satisfies $\psi_i(x,y_1) \land \varphi_i(x) \land
  y_1 \geq 0$ for some $i \in I$. Observe that $y_1$ only
  appears in polynomials on the right-hand side of
  divisibilities.
\end{example}

The process sketched in the example can be applied in general to BIL sentences
sequentially starting from the rightmost quantified $y_i$. At each step, the
size of the formula is at most squared. In what follows, it will be convenient to deal with a single polyadic $\gcd$ instead of nested binary ones. Thus, using associativity of $\gcd$ and
pushing coefficients inwards --- i.e. using the equivalence $a \cdot \gcd(x,y)
\equiv \gcd(ax, ay)$ for $a \in \mathbb{N}$ --- we finally obtain a sentence:
\begin{equation}\label{eqn:bil-gcd}
  \forall x_1 \dots \forall x_n \bigvee_{i \in I} \bigwedge_{j \in L_i}
  (\gcd(\{f'_{j,k}(\vec{x})\}_{k=1}^{K_j}) \divs g'_j(\vec{x}))
  \land \varphi'_i(\vec{x})
\end{equation}
where $|L_i|$, $|K_j|$, and the coefficients may all be doubly-exponential
in the number $m$ of removed variables, due to iterated squaring.

\subparagraph{Eliminating the gcd operator.} In this next step, our goal is to obtain an \apad sentence from Equation \eqref{eqn:bil-gcd}. Recall that \apad ``natively'' allows for negated divisibility constraints. (That is, without having to encode them using Lechner's trick.) Hence, to remove expressions in terms of $\gcd$ from
Equation \eqref{eqn:bil-gcd}, we can use the following identity:
\[
  \gcd(f_1(\vec{x}), \dots, f_n(\vec{x})) \divs g(\vec{x})
  \iff \forall d \left(\bigwedge_{i=1}^n d \divs f_i(\vec{x})\right)
  \rightarrow d \divs g(\vec{x}).
\]
This substitution results in a constant blowup of the size of the sentence.
The above method gives us a sentence $\forall \vec{x} \forall d
\psi(\vec{x},d)$, where $\psi(\vec{x},d)$ is a quantifier-free PAD formula. To
summarize:

\begin{lemma}\label{lem:bil2univ}
  For any BIL sentence $\varphi = \forall x_1 \dots \forall x_n \exists y_1
  \dots \exists y_m \bigvee_{i \in I} \varphi_{i}(\vec{x},\vec{y})$ we can
  construct an \apad sentence $\psi = \forall x_1 \dots \forall x_n \forall d
  \bigvee_{i \in I} \psi_{i}(\vec{x},d)$ such that:
  $\varphi$ is true if and only if $\psi$ is true and
  for all $i \in I$, $|\psi_i| \leq |\varphi_i|^{2^m}$.
  The construction is realizable in time $\mathcal{O}(|\varphi|^{2^m})$.
\end{lemma}

To prove \Cref{thm:bildec}, the following small-model results for purely existential PAD formulas and BIL will be useful.

\begin{theorem}[{\cite[Theorem 14]{LechnerOW15}}]\label{thm:epad-soln}
  Let $\varphi(x_1,\dots,x_n)$ be a \epad formula. If $\varphi$ has
  a solution then it has a solution $(a_1,\dots,a_n) \in \Z^n$ with
  the bitsize of each $a_i$ bounded by $|\varphi|^{\mathrm{poly}(n)}$.
\end{theorem} 
\begin{corollary}\label{cor:bil-small-model}
    Let $\forall x_1 \dots \forall x_n \varphi(x_1,\dots,x_n)$ be a BIL sentence. If $\lnot \varphi$ has a solution then it has a solution $(a_1,\dots,a_n) \in \mathbb{Z}^n$ with the bitsize of each $a_i$ bounded by \(|\varphi|^{2^m \mathrm{poly}(n+1)}\).
\end{corollary}
\begin{proof}
  Using \Cref{lem:bil2univ}, we
  translate the
  BIL sentence to $\forall x_1 \dots \forall
  x_n \forall d \psi(\vec{x},d)$, where the latter is an \apad sentence.
  Then, using \Cref{thm:epad-soln}, we get that the \epad formula $\lnot
  \psi(\vec{x},d)$ admits a solution if and only if it has one with bitsize bounded by
  $|\psi|^{\mathrm{poly}(n+1)}$. Now, from \Cref{lem:bil2univ} we have that
  $|\psi|$ is bounded by $|\varphi|^{2^m}$. Hence, 
  we get that the bitsize of a
  solution is bounded by:
  \(|\varphi|^{2^m \mathrm{poly}(n+1)}\).
\end{proof}
We are now ready to prove the theorem.
\begin{proof}[Proof of \Cref{thm:bildec}]
  As in the proof of \Cref{cor:bil-small-model}, we translate the BIL sentence to $\forall x_1 \dots \forall
  x_n \forall d \psi(\vec{x},d)$. Note
  that our algorithm thus far runs in time:
  \(\mathcal{O}\left(|\varphi|^{2^m}\right)\). By \Cref{cor:bil-small-model}, if $\lnot \psi(\vec{x},d)$ has a solution then it has
  one encodable in binary using a doubly exponential amount of bits with
  respect to the size of the input BIL sentence. The naive guess-and-check
  decision procedure applied to $\lnot\psi(\vec{x},d)$ gives us a $\coNTWOEXP$
  algorithm for BIL sentences. Indeed, after computing $\psi(\vec{x},d)$ and
  guessing a valuation, checking it satisfies $\lnot\psi$ takes polynomial
  time in the bitsize of the valuation and $|\psi|$, hence doubly
  exponential time in $|\varphi|$.
\end{proof}


\section{Succinct One-Counter Automata with Parameters}
\label{sec:prelims}
We now define OCA with parameters and recall
some basic properties. The concepts and observations we
introduce here are largely taken from~\cite{hkow09} and the exposition in
\cite{lechner16}. 

A \emph{succinct parametric one-counter automaton (SOCAP)} is a tuple
$\Au=(Q,T,\delta, X)$, where $Q$ is a finite set of states, $X$ is a finite set of parameters,
$T \subseteq Q \times Q$ is a finite set of transitions and 
$\delta: T \to \Op$ is a function that associates an operation to every transition. The set $\Op = \cu \uplus \pu \uplus
\zt \uplus \pt$ is the union of: 
\emph{Constant Updates} $\cu \defeq \{+a : a \in \Z\}$,
      \emph{Parametric Updates} $\pu \defeq \{Sx: S \in \{+1,-1\}, x \in
      X\}$,
\emph{Zero Tests} $\zt \defeq \{=0 \}$, and
      \emph{Parametric Tests} $\pt \defeq \{=x, \geq x : x \in X \}$.
%
We denote by ``$=0$'' or ``{$=x$}'' an \emph{equality test} between
the value of the counter and zero or the value of $x$ respectively; by ``$\geq x$'', a \emph{lower-bound
test} between the values of the counter and $x$.
A \emph{valuation} $V : X \to \N$ assigns to every parameter a natural number.
We assume $\cu$ are encoded in binary, hence the S in SOCAP.
We omit ``parametric'' if $X= \emptyset$ and often write $q \xrightarrow{\op}
q'$ to denote $\delta(q,q')=\op$.

A \emph{configuration} is a pair $(q,c)$ where $q \in Q$ and $c \in \N$ is the
\emph{counter value}.  Given a valuation $V : X \to \N$ and a configuration
$(q_0,c_0)$, a \emph{$V$-run from $(q_0,c_0)$} is a sequence $\rho=
(q_0,c_0)(q_1,c_1)\dots$ such that for all $i \geq 0$ the following hold: $q_i
\xrightarrow{op_{i+1}} q_{i+1}$; $c_i = 0$, $c_i = V(x)$, and $c_i \geq V(x)$,
if $\delta(q_i,q_{i+1})$ is ``$=0$'', ``$=x$'', and ``$\geq x$'',
respectively; and $c_{i+1}$ is obtained from $c_i$ based on the counter
operations. That is, $c_{i+1}$ is $c_i$ if $\delta(q_i,q_{i+1}) \in \left(\zt
\cup \pt\right)$; $c_i + a$ if $\delta(q_i,q_{i+1}) = +a$; $c_i + S \cdot V(x)$ if
$\delta(q_i,q_{i+1}) = Sx$.
We say $\rho$ \emph{reaches} a state $q_f \in Q$ if there exists $j \in
\mathbb{N}$, such that $q_j= q_f$. Also, $\rho$ reaches or visits a set of
states $F \subseteq Q$ iff $\rho$ reaches a state $q_f \in F$. If $V$ is
clear from the context we just write run instead of $V$-run.

The \emph{underlying (directed)
graph} of $\Au$ is $G_\Au=(Q,T)$. A $V$-run $\rho = (q_0,c_0) (q_1,c_1) \dots$
in $\Au$ \emph{induces} a path $\pi = q_0 q_1 \dots$ in $G_\Au$.  We assign
weights to $G_\Au$ as follows: For $t \in T$, $\weight(t)$ is $0$ if
$\delta(t) \in \zt \cup \pt$; $a$ if $\delta(t) = +a$; and $S\cdot V(x)$ if
$\delta(t) = Sx$.  We extend the $\weight$ function to finite paths
in the
natural way. Namely, 
we set $\weight(q_0 \dots q_n) \defeq \sum_{i=0}^{n-1}
\weight(q_i,q_{i+1})$. 

\subparagraph{Synthesis problems.}
\emph{The synthesis problem} asks, given a SOCAP $\Au$, a state $q$ and an
$\omega$-regular property $p$, whether there exists a valuation $V$ such that 
all infinite $V$-runs from $(q,0)$ satisfy $p$.
We focus on the following classes of $\omega$-regular properties. Given a set
of target states $F \subseteq Q$ and an infinite run $\rho= (q_0, c_0) (q_1,
c_1)\dots$ we say $\rho$ satisfies:
\begin{itemize}
    \item
the \emph{reachability} condition if $q_i \in F$ for some $i \in \N$;
    \item
the \emph{B\"uchi} condition if $q_i \in F$ for infinitely many $i \in \N$;
    \item
the \emph{coB\"uchi} condition if $q_i \in F$ for finitely many $i \in \N$ only;
    \item
the \emph{safety} condition if $q_i \not\in F$ for all $i \in \N$;
    \item
the \emph{linear temporal logic} (LTL) formula $\varphi$ over a set of
atomic propositions $P$ --- and with respect to a labelling function $f: Q \to
2^P$ --- if $f(q_0) f(q_1) \dots \models \varphi$.\footnote{See,
e.g.,~\cite{bk08} for the classical semantics of LTL.}
\end{itemize}

We will decompose the synthesis problems into reachability sub-problems. It
will thus be useful to recall the following connection between reachability
(witnesses) and graph flows.


\subparagraph*{Flows.} 
For a directed graph $G = (V,E)$, we denote the set of immediate successors of
$v \in V$ by $vE\coloneqq \{w \in V \mid (v,w) \in E\}$ and the immediate
predecessors of $v$ by $Ev$, defined analogously. 
An \emph{$s$--$t$ flow} is a mapping $f : E \to \mathbb{N}$ that satisfies
flow conservation:
\(
    \forall v \in V \setminus \{s,t\} :
    \sum_{u \in Ev} f(u,v) = \sum_{u \in vE} f(v,u).
\)
That is, the total incoming flow equals the total outgoing flow for all but
the source and the target vertices. We  then define the \emph{value} of a flow
$f$ as:
\(
    |f| \defeq \sum_{v \in sE} f(s,v) - \sum_{u \in Es} f(u,s) .
\)
We denote by $\emph{support}(f)$ the set $\{e \in E \mid f(e) > 0\}$ of edges
with non-zero flow.  A cycle in a flow $f$ is a cycle in the sub-graph induced
by $\mathrm{support}(f)$. For weighted graphs, we define $\weight(f) \defeq
\sum_{e \in E} f(e) \weight(e)$. 

\subparagraph*{Path flows.}
Consider a path $ \pi= v_0 v_1 \dots$ in $G$. We denote by $f_\pi$ its Parikh
image, i.e. $f_\pi$ maps each edge $e$ to the number of times $e$ occurs in
$\pi$. A flow $f$ is called a \emph{path flow} if there exists a path $\pi$
such that $f = f_\pi$.
Finally, we observe that an $s$--$t$ path flow $f$ in $G$ induces a $t$--$s$ path flow
$f'$ with $f'(u,v) = f(v,u)$, for all $(u,v) \in E$, in the skew transpose of $G$.

\section{Encoding Synthesis Problems into the BIL Fragment}\label{sec:encoding}

\label{sec:synthOCA}
In this section, we prove that all our synthesis problems
are decidable.
More precisely, we establish the following complexity upper bounds.

\begin{theorem} \label{thm-synthreach}
    The reachability, B\"uchi, coB\"uchi, and safety synthesis
    problems for succinct one-counter automata with parameters
    are all decidable in \NTWOEXP.
\end{theorem}
The idea is as follows: we focus on the coB\"uchi synthesis problem and reduce
its complement to the truth value of a BIL sentence. To do so, we follow
Lechner's encoding of the complement of the B\"uchi synthesis problem into
\aerpadplus \cite{lechner15}. The encoding heavily
relies on an encoding for (existential) reachability from~\cite{hkow09}. We
take extra care to obtain a BIL sentence instead of an \aerpadplus one as
Lechner originally does.

It can be shown that the other synthesis problems reduce to the coB\"uchi one in polynomial time. The
corresponding bounds thus follow from the one for coB\"uchi synthesis. The proof of the following lemma is given in the long version of the paper.

\begin{lemma} \label{lem_synth_red}
  The reachability, safety and the B\"uchi synthesis problems can be reduced
  to the coB\"uchi synthesis problem in polynomial time.
\end{lemma}

Now the cornerstone of our reduction from the complement of the coB\"uchi
synthesis problem to the truth value of a BIL sentence is an encoding of
\emph{reachability certificates} into \aerpad formulas which are ``almost'' in
BIL.
%
In the following subsections we will focus on a SOCAP $\Au = (Q,T,\delta,X)$
with $X = \{x_1,\dots,x_n\}$ and often write $\vec{x}$ for $(x_1,\dots,x_n)$.
We will prove that the existence of a $V$-run from $(q,c)$ to $(q',c')$ can be
reduced to the satisfiability problem for such a formula.
\begin{proposition} \label{prop_reach_pbil}
    Given states $q,q'$, one can construct in deterministic exponential time in $|\Au|$ a
    PAD formula:
    \(
    \reach^{(q,q')}(\vec{x},a,b) = \exists \vec{y} \bigvee_{i \in I}
      \varphi_i(\vec{x},\vec{y}) \land \psi_i(\vec{y},a,b) \land
      \vec{y} \geq \vec{0}
    \)
    such that $\forall \vec{x} \exists \vec{y} \bigvee_{i \in I}
    \varphi_i(\vec{x},\vec{y}) \land \vec{y} \geq 0$ is a BIL
    sentence, the $\psi_i(\vec{y},a,b)$ are quantifier-free PA formulas, and
    additionally:
    \begin{itemize}
      \item
    a valuation $V$ of $X\cup\{a,b\}$ satisfies
        $\reach^{(q,q')}$ iff there is a $V$-run from
        $(q,V(a))$ to $(q',V(b))$;
       \item
        the bitsize of constants in $\reach^{(q,q')}$ is of polynomial
         size in $|\Au|$;
       \item
         $|\reach^{(q,q')}|$ is at most exponential with respect to
         $|\Au|$; and 
       \item
         the number of $\vec{y}$ variables is polynomial with respect to $|\Au|$.
    \end{itemize}
\end{proposition}
Below, we make use of this proposition to prove \Cref{thm-synthreach}. Then, we prove some auxiliary results in \Cref{sec:reach-cert} and, in \Cref{sec:together}, we present a sketch of our proof of \Cref{prop_reach_pbil}.

We will argue that
$\forall\vec{x} \exists a \exists b \reach^{(q,q')}(\vec{x},a,b)$
can be transformed into an equivalent BIL sentence. 
Note that for this to
be the case it suffices to remove the $\psi_i(\vec{y},a,b)$ subformulas.
Intuitively, since these are quantifier-free PA formulas, their set of
satisfying valuations is \emph{semi-linear} (see, for
instance,~\cite{haase18}). Our intention is to remove the
$\psi_i(\vec{y},a,b)$ and replace the occurrences of $\vec{y},a,b$ in the rest of
$\reach^{(q,q')}(\vec{x},a,b)$ with linear polynomials ``generating''
their set of solutions. This is formalized below.

\subparagraph{Affine change of variables.} 
Let $\vec{A} \in \mathbb{Z}^{m \times n}$
be an integer matrix of size $m \times n$ of rank $r$,
and $\vec{b} \in \mathbb{Z}^m$.
Let $\vec{C} \in \mathbb{Z}^{p \times n}$ be
an integer matrix of size $p \times n$ such that
$\begin{psmallmatrix} \vec{A}\\ \vec{C} \end{psmallmatrix}$ has
rank $s$, and $\vec{d} \in \mathbb{Z}^p$. We write $\mu$ for the maximum absolute value
of an $(s-1) \times (s-1)$ or $s \times s$ sub-determinant
of the matrix $\begin{psmallmatrix} \vec{A} & \vec{b}\\ \vec{C} & \vec{d} 
\end{psmallmatrix}$ that incorporates at least $r$ rows from
$\begin{pmatrix} \vec{A} & \vec{b} \end{pmatrix}$. 
\begin{theorem}[From~\cite{vzgs78}]\label{thm:gathking}
    Given integer matrices $\vec{A} \in \mathbb{Z}^{m \times n}$ and
    $\vec{C} \in \mathbb{Z}^{p \times n}$, integer vectors
    $\vec{b} \in \mathbb{Z}^m$ and $\vec{d} \in \mathbb{Z}^p$,
    and $\mu$ defined as above,
    there exists a finite set $I$, a collection of $n \times (n - r)$
    matrices $\vec{E}^{(i)}$, and $n \times 1$ vectors
    $\vec{u}^{(i)}$, indexed by $i \in I$, all with integer 
    entries bounded by $(n+1)\mu$ such that:
    \(
      \{\vec{x} \in \mathbb{Z}^n : \vec{A}\vec{x} = \vec{b} \land 
      \vec{C}\vec{x} \geq \vec{d}\}
      = \bigcup_{i \in I}
      \{\vec{E}^{(i)}\vec{y} + \vec{u}^{(i)} : \vec{y} \in \mathbb{Z}^{n-r},
      \vec{y} \geq \vec{0}\}.
    \)
\end{theorem}


We are now ready to prove \Cref{thm-synthreach}.
\begin{proof}[Proof of \Cref{thm-synthreach}]
We will first prove that the complement of the coB\"uchi synthesis problem can
be encoded into a BIL sentence. Recall that the complement of the coB\"uchi
synthesis problem asks: given a SOCAP $\Au$ with parameters $X$, for all
valuations does there exist an infinite run from a given configuration
$(q,0)$, that visits the target set $F$ infinitely many times. Without loss of generality, we assume 
that the automaton has no parametric tests as they can be simulated using parametric
 updates and zero tests.

The idea is to
check if there exists a reachable ``pumpable cycle'' containing one of the
target states. Formally, given the starting configuration $(q,0)$, we want to
check if we can reach a configuration $(q_f, k)$, where $q_f \in F$ and $k \geq 0$ and then we want to reach $q_f$ again via a pumpable
cycle.  This means that starting from $(q_f,k)$ we reach the configuration
$(q_f,k)$ again or we reach a configuration $(q_f, k')$ with $k'\geq k$
without using zero-test transitions. Note that reachability while avoiding
zero tests is the same as reachability in the sub-automaton obtained after
deleting all the zero-test transitions. We write $\reachnt$ for the $\reach$
formula constructed for that sub-automaton as per \Cref{prop_reach_pbil}.
The above constraints can
be encoded as a formula $\cob(\vec{x}) = \exists k \exists k' \bigvee_{q_f
\in F} \zeta(\vec{x},k,k')$ where the subformula $\zeta$ is:
\(
  (k \leq k') 
  \land \reach^{(q,q_f)}(\vec{x},0,k)
  \land \left( \reachnt^{(q_f,q_f)}(\vec{x},k,k') \lor
  \reach^{(q_f,q_f)}(\vec{x},k,k) \right).
\)
Finally, the formula
$\cob(\vec{x})$ will look as follows:
\[
    \exists \vec{y} \exists k \exists k'
     \bigvee_{i \in I} \bigwedge_{j \in J_i} 
    \left(f_j(\vec{x}) \divs g_j(\vec{x},\vec{y})\right)
    \land \varphi_i(\vec{x})\land \psi_i(\vec{y},k, k') \land \vec{y} \geq \vec{0}   
\]
where, by \Cref{prop_reach_pbil}, the $\varphi_i(\vec{x})$ are quantifier-free
PA formulas over $\vec{x}$ constructed by grouping all the quantifier-free PA
formulas over $\vec{x}$. Similarly, we can construct $\psi_i(\vec{y}, k, k')$
by grouping all the quantifier free formulas over $\vec{y}, k$ and $k'$. Now, 
we use the affine change of variables
to remove the formulas $\psi_i(\vec{y},k, k')$. Technically, the free
variables from the subformulas $\psi_i$ will be replaced in all other
subformulas by linear polynomials on newly introduced variables $\vec{z}$.
Hence, the final formula $\cob(\vec{x})$ becomes: 
\[
\exists \vec{z} \bigvee_{i \in I'} \bigwedge_{j \in J_i}  \left(f_j(\vec{x}) \divs g_j(\vec{x},\vec{z})\right) \land \varphi_i(\vec{x}) \land \vec{z} \geq \vec{0}.
\] 

Note that, after using the affine change of variables, the number of $\vec{z}$
variables is bounded by the number of old existentially quantified variables
($\vec{y},k,k'$). However, we have introduced exponentially many new
disjuncts.\footnote{Indeed, because of the bounds on the entries of the
matrices and vectors, the cardinality of the set $I$ is exponentially
bounded.}

By construction, for a valuation $V$ there is an infinite $V$-run in $\Au$
from $(q,0)$ that visits the target states infinitely often iff
$\cob(V(\vec{x}))$ is true. Hence, $\forall \vec{x} (\vec{x} < 0 \lor
\cob(\vec{x}))$ precisely
encodes the complement of the coB\"uchi synthesis problem. Also, note that it
is a BIL sentence since the subformulas (and in particular the divisibility constraints)
come from our usage of \Cref{prop_reach_pbil}.
Now, the number of $\vec{z}$ variables, say $m$, is bounded by the number of
$\vec{y}$ variables before the affine change of variables which is polynomial
with respect to $|\Au|$ from  \Cref{prop_reach_pbil}. Also, the bitsize of
the constants in $\cob$ is polynomial in $|\Au|$ though the size of the
formula is exponential in $|\Au|$. Now,
using \Cref{lem:bil2univ},
we construct an \apad sentence $\forall \vec{x} \forall d \psi(\vec{x},d)$
from $\forall \vec{x} (\vec{x} < 0 \lor \cob(\vec{x}))$. By \Cref{cor:bil-small-model},
$\lnot\psi$ admits a solution of bitsize bounded by:
\(
\exp({\ln(|\cob|)2^m \mathrm{poly}(n+1)})= \exp({|\Au| \cdot
2^{\mathrm{poly}(|\Au|)} \mathrm{poly}(n+1)}),
\)
which is doubly exponential in the size of $|\Au|$. As in the
proof of \Cref{thm:bildec}, a guess-and-check
algorithm for $\lnot\psi$ gives us the desired
\NTWOEXP{} complexity result for the coB\"uchi synthesis problem. By
\Cref{lem_synth_red}, the other synthesis problems have the same
complexity.
\end{proof}

In the sequel we sketch our proof of \Cref{prop_reach_pbil}.

\subsection{Reachability certificates}\label{sec:reach-cert}
We presently recall the notion of reachability certificates from~\cite{hkow09}.
Fix a SOCAP $\Au$ and a valuation $V$. A flow $f$ in $G_\Au$ is a \emph{reachability
certificate} for two configurations $(q,c),(q',c')$ in $\Au$ if there is a $V$-run from $(q,c)$ to $(q',c')$ that induces a path $\pi$ such that $f = f_\pi$
and one of the following holds:
\emph{(type 1)} $f$ has no positive-weight cycles,
\emph{(type 2)} $f$ has no negative-weight cycles, or
\emph{(type 3)} $f$ has a positive-weight cycle that can be taken from $(q,c)$ and
    a negative-weight cycle that can be taken to $(q',c')$.

In the sequel, we will encode the conditions from the following result into a
PAD formula so as to accommodate parameters. Intuitively, the proposition
states that there is a run from $(q,c)$ to $(q',c')$ if and only if there is
one of a special form: a decreasing prefix (type 1), a positive cycle leading
to a plateau followed by a negative cycle (type 3), and an increasing suffix
(type 2). Each one of the three sub-runs could in fact be an empty run.

\begin{proposition}[{\cite[Lemma 4.1.14]{christoph2012a}}]\label{pro:certificates} 
    If $(q', c')$ is reachable from $(q, c)$ in a SOCAP with $X = \emptyset$
    and without zero tests then there is a run $\rho = \rho_1\rho_2\rho_3$
    from $(q, c)$ to $(q', c')$, where $\rho_1$, $\rho_2$,
    $\rho_3$, each have a polynomial-size reachability certificate of type 1,
    3 and 2, respectively.
\end{proposition}

\subparagraph{Encoding the certificates.}
Now, we recall the encoding for the reachability certificates proposed
by Lechner~\cite{lechner15,lechner16}. Then, we highlight the changes
necessary to obtain the required type of formula.
We begin with type-1 and type-3 certificates.
\begin{lemma}[{From~\cite[Lem.~33 and Prop.~36]{lechner16}}]
    Suppose $\Au$ has no zero tests and let $t \in \{1,2,3\}$.
    Given states $q,q'$, one can construct in deterministic exponential time
    the existential PAD formula $\Phi^{(q,q')}_t(\vec{x},a,b)$. Moreover, a valuation
    $V$ of $X \cup \{a,b\}$ satisfies $\Phi_t^{(q,q')}(\vec{x},a,b)$ iff there
    is a $V$-run from $(q,V(a))$ to $(q',V(b))$ that induces a path
    $\pi$ with $f_\pi$ a type-$t$ reachability certificate.
\end{lemma}
The formulas $\Phi^{(q,q')}_t$ from the result above look as follows:
\[
    \bigvee_{i \in I} \exists \vec{z}
    \bigwedge_{j \in J_i} m_j(\vec{x}) \divs z_j
    \land (m_j(\vec{x}) > 0 \leftrightarrow z_j > 0)
    \land \varphi_i(\vec{x})
    \land \psi_i(\vec{z},a,b)
    \land \vec{z} \geq \vec{0}
\]
where $|I|$ and the size of each disjunct are exponential.\footnote{Lechner~\cite{lechner15}
  actually employs a symbolic encoding of the Bellman-Ford algorithm to get
  polynomial disjuncts
  in her formula. However, 
  a na\"ive encoding --- while exponential --- yields the formula we present here and
streamlines its eventual transformation to BIL.}
Further, all the $\varphi_i$ and $\psi_i$ are
quantifier-free PA formulas and the $m_j(\vec{x})$ are all either $x$, $-x$,
or $n \in \mathbb{N}_{>0}$.

We observe that the constraint $(m_j(\vec{x}) > 0 \leftrightarrow z_j > 0)$
regarding when the variables can be $0$, can be pushed into a further
disjunction over which subset of $X$ is set to $0$. In one case the
corresponding $m_j(\vec{x})$'s and $z_j$'s are replaced by $0$, in the
remaining case we add to $\varphi_i$ and $\psi_i$ the constraints $z_j > 0$
and $m_j(\vec{x}) > 0$ respectively. 
%
We thus obtain formulas $\Psi_t^{(q,q')}(\vec{x},a,b)$ with the following
properties.
\begin{lemma} \label{lem-type1}
  Suppose $\Au$ has no zero tests and let $t \in \{1,2,3\}$.
    Given states $q,q'$, one can construct in deterministic exponential time a
    PAD formula
    \(
    \Psi_t^{(q,q')}(\vec{x},a,b) = \exists \vec{y} \bigvee_{i \in I}
      \varphi_i(\vec{x},\vec{y}) \land \psi_i(\vec{y},a,b) \land
      \vec{y} \geq \vec{0}
    \)
    s.t. $\forall \vec{x} \exists \vec{y} \bigvee_{i \in I}
    \varphi_i(\vec{x},\vec{y}) \land \vec{y} \geq \vec{0}$ is a BIL
    sentence, the $\psi_i(\vec{y},a,b)$ are quantifier-free PA formulas, and
    additionally:
   \begin{itemize}
    \item a valuation $V$ of $X \cup \{a,b\}$ satisfies $\Psi^{(q,q')}_t$ iff
      there is a $V$-run from $(q,V(a))$ to $(q',V(b))$ that induces a path
      $\pi$ such that $f_\pi$ is a type-$t$ reachability certificate,
    \item the bitsize of constants in $\Psi_t^{(q,q')}$ is of polynomial size
      in $|\Au|$,
    \item $|\Psi_t^{(q,q')}|$ is at most exponential with respect to $|\Au|$, and 
    \item the number of $\vec{y}$ variables is polynomial with respect to $|\Au|$.
    \end{itemize}
\end{lemma}

\subsection{Putting everything together}\label{sec:together}
In this section, we combine the results from the previous subsection to
construct $\reach$ for \cref{prop_reach_pbil}. 
%
The construction, in full detail, and a formal proof that $\reach$ enjoys the claimed properties are
given in the long version of this paper.
First, using \Cref{pro:certificates}
and the lemmas above, we define a formula
$\reachnt^{(q,q')}(\vec{x},a,b)$ that is satisfied by a valuation $V$ of $X
\cup \{a,b\}$ iff there is a $V$-run from $(q, V(a))$ to $(q', V(b))$ without
any zero-test transitions. To do so, we use formulas for the sub-automaton
obtained by removing from $\Au$ all zero-test transitions.
Then, the formula $\reach^{(q,q')}(\vec{x},a,b)$ expressing general reachability can
be defined by taking a disjunction over all orderings on the zero tests.
In other words, for each enumeration of zero-test transitions we take the
conjunction of the intermediate $\reachnt$ formulas as well as $\reachnt$
formulas from the initial configuration and to the final one.

Recall
that for any LTL formula $\varphi$ we can construct a \emph{universal coB\"uchi automaton} of exponential size in $|\varphi|$~\cite{bk08,kv05}. (A universal coB\"uchi automaton accepts a word $w$ if all of its infinite runs on $w$ visit $F$ only finitely often. Technically, one can construct such an automaton for $\varphi$ by constructing a B\"uchi automaton for $\lnot \varphi$ and ``syntactically complementing'' its acceptance condition.) By considering the product of
this universal coB\"uchi automaton and the given SOCAP, the LTL synthesis problem reduces to coB\"uchi synthesis.
\begin{corollary}\label{cor:ltl-synth}
    The LTL synthesis problem for succinct one-counter automata with parameters is decidable in \NTHREEEXP.
\end{corollary}

\section{One-Counter Automata with Parametric Tests}

\label{sec:ocapt}

In this section, we introduce a subclass of SOCAP where only the tests are parametric. The updates are non-parametric and assumed to be given in unary. Formally, \emph{OCA with parametric tests (OCAPT)} allow for constant updates of the form $\{+a : a \in \{-1,0,1\}\}$ and zero and parametric tests. However, $\pu = \emptyset$.

We consider the synthesis problems for OCAPT. Our main result in this section are better complexity upper bounds than for general SOCAP. \Cref{lem_synth_red} states that all the synthesis problems reduce to the coB\"uchi synthesis problem for SOCAP. Importantly, in the construction used to prove \Cref{lem_synth_red}, we do not introduce parametric updates. Hence, the reduction also holds for OCAPT. This allows us to focus on the coB\"uchi synthesis problem --- the upper bounds for the other synthesis problems follow.

\begin{theorem} \label{thm_synthreach_ocapt}
    The coB\"uchi, B\"uchi and safety synthesis problems for OCAPT are in \PSPACE; the reachability synthesis problem, in
    $\NP^{\coNP} = \NP^{\NP}$. 
\end{theorem}
To prove the theorem, we follow an idea from~\cite{bqs19} to encode parameter valuations of OCAPT into words accepted by an alternating two-way automaton. Below, we give the proof of the theorem assuming some auxiliary results that will be established in the following subsections.

\begin{proof}
   In \Cref{OCA to A2A}, we reduce the coB\"uchi synthesis problem to the non-emptiness problem for alternating two-way automata. Hence, we get the \PSPACE{} upper bound. Since the B\"uchi and the safety synthesis problems reduce to the coB\"uchi one (using \Cref{lem_synth_red}) in polynomial time, these are also in \PSPACE. 

Next, we improve the complexity upper bound for the reachability synthesis problem from $\PSPACE$ to $\NP^\NP$. In \Cref{bounded-parameter} we will prove that if there is a valuation $V$ of
   the parameters such that all infinite $V$-runs reach $F$ then we can assume that $V$
   assigns to each $x \in X$ a value at most exponential. Hence, we can guess their binary encoding and store it using a polynomial number of bits.
   Once we have guessed $V$
   and replaced all the $x_i$ by $V(x_i)$,
   we obtain a  non-parametric one counter automata $\Au'$ with $X= \emptyset$ and we ask whether all infinite runs reach $F$. We will see in \Cref{prop_synthreach_nopara} that
   this problem is in $\coNP$. The claimed complexity upper bound for the reachability synthesis problem follows.
\end{proof}
Using a similar idea to \Cref{cor:ltl-synth}, we reduce the LTL synthesis problem to the coB\"uchi one and 
we obtain the following.
\begin{corollary}\label{cor:ltl-synth2}
    The LTL synthesis problem for OCAPT is in \EXPSPACE.
\end{corollary}
\subsection{Alternating two-way automata} 
Given a finite set $Y$, we denote by $\mathbb{B}^{+}(Y)$ the set of
positive Boolean formulas over $Y$, including $\true$ and $\false$. A
subset $Y' \subseteq Y$ satisfies $\beta \in \mathbb{B}^{+}(Y)$, written
$Y' \models \beta$, if $\beta$ is evaluated to $\true$ when substituting
$\true$ for every element in $Y'$, and $\false$ for every element in
$Y\setminus Y'$. In particular, we have $\emptyset \models \true$. 

We can now define an \emph{alternating two-way automaton} (A2A, for short) as a tuple 
$\T = (S, \Sigma, s_{in}, \Delta, S_f)$,
where $S$ is a finite set of states,
$\Sigma$ is a finite alphabet,
$s_{in} \in S$ is the initial state,
$S_f \subseteq S$ is the set of accepting states, and
$\Delta \subseteq S \times (\Sigma \cup \{\first?\}) \times \mathbb{B}^{+}(S \times 
\{+1,0,-1\})$ is the finite transition relation. 
%
The $+1$ intuitively means that the head moves to the right;
$-1$, that the head moves to the left; $0$, that it stays at the
current position. Furthermore, transitions are labelled by
Boolean formulas over successors which determine whether the current
run branches off in a non-deterministic or a universal
fashion.

A \emph{run (tree)} $\gamma$ of $\T$ on an infinite word 
$w = a_0 a_1 \dots \in \Sigma^w$
from $n \in \mathbb{N}$ is a (possibly infinite) rooted tree whose vertices are labelled with elements
in $S \times \N$ and such that it satisfies the following
properties. 
The root of $\gamma$ is labelled
by $(s_{in}, n)$.
Moreover, for every vertex labelled by $(s, m)$ with
$k \in \N$ children labelled by $(s_1 ,n_1) , \dots, (s_k, n_k)$, there
is a transition $(s, \sigma, \beta) \in \Delta$ such that, the set $\{(s_1, n_1-m), \dots, (s_k, n_k-m)\} \subseteq S \times \{+1, 0, -1\}$ satisfies $\beta$. Further $\sigma = \first?$ implies $m=0$, and $\sigma \in \Sigma$ implies $a_m = \sigma$.

%
A run is \emph{accepting} if all of its infinite branches
contain infinitely many 
labels from $S_f
\times \mathbb{N}$.
%
%
%
%
The \emph{language of $\T$} is $L(\T) \defeq \{w \in \Sigma^{\omega}
  \mid \exists \text{ an accepting run of } \T \text{ on } w\text{ from } 0\}$.
The \emph{non-emptiness problem for A2As} asks,
given an A2A
$\T$ and $n \in \mathbb{N}$, whether
$L(\T) \not = \emptyset$.
\begin{proposition} [From \cite{Serre06}] \label{A2A} 
   Language emptiness for A2As is in \PSPACE.
\end{proposition}

In what follows, from a given OCAPT $\Au$ we will build an A2A $\T$
such that $\T$ accepts precisely those words which correspond to a
valuation $V$ of
$X$ under which all infinite runs satisfy the coB\"uchi condition. Hence,
the corresponding synthesis problem for $\Au$ reduces to checking
non-emptiness of $\T$.

\subsection{Transformation to alternating two-way automata}
Following~\cite{bqs19}, we encode a valuation $V : X \to \mathbb{N}$ as an
infinite \emph{parameter word} $w =  a_0 a_1 a_2 \dots$ over the alphabet
$\Sigma= X \cup \{\square\}$ such that $a_0= \square$ and, for every $x \in X$,
there is exactly one position $i \in \N$ such that $a_i=x$.
We write $w(i)$ to
denote its prefix $a_0 a_1 \dots a_i$ up to the letter $a_i$. By
$|w(i)|_{\square}$, we denote the number of occurrences of
$\square$ in $a_1 \dots a_i$. (Note that we ignore $a_0$.)
Then, a parameter word $w$ determines a valuation $V_w: x \mapsto |w(i)|_{\square}$
where $a_i = x$.
%
%
We denote the set of all parameter words over $X$ by $W_X$.

From a given OCAPT $\Au = (Q,T,\delta,X)$, a starting configuration $(q_0, 0)$ and a set of target states $F$, we will now construct an A2A $\T = (S,\Sigma,s_{in},\Delta,S_f)$ that accepts words $w \in W_X$ such that,
under the valuation $V= V_w$, all infinite runs from $(q_0, 0)$ visit $F$ only finitely many times.
\begin{proposition} \label{OCA to A2A}
    For all OCAPT $\Au$ there is an A2A $\T$ with $|\T| =
    |\Au|^{\mathcal{O}(1)}$ and $w \in L(\T)$ if and only if all infinite $V_w$-runs of
    $\Au$ starting from $(q_0,0)$ visit $F$ only finitely many times.
\end{proposition}
The construction is based on the A2A built in~\cite{bqs19}, although we make
more extensive use of the alternating semantics of the automaton. To capture the coB\"uchi condition, we simulate a \emph{safety copy} with the target states as ``non-accepting sink'' (states having a self-loop and no other outgoing transitions) inside $\T$. Simulated accepting runs of $\Au$ can ``choose'' to enter said safety copy once they are sure to never visit $F$ again. Hence, for every state $q$ in $\Au$, we have two copies of the state in $\T$: $q' \in S$ representing $q$ normally and $q'' \in S$ representing $q$ from the safety copy. 
Now the idea is to encode runs of $\Au$ as branches of run trees of $\T$ on parameter words $w$ by letting sub-trees $t$ whose root is labelled with $(q',i)$ or $(q'',i)$ correspond to
the configuration $(q,|w(i)|_\square)$ of $\Au$.
If $t$ is accepting, it will serve as a witness that all infinite runs of $\Au$ from $(q,|w(i)|_\square)$
satisfy the coB\"uchi condition.

We present the overview of the construction below with some intuitions. A detailed proof of~\Cref{OCA to A2A} is given in the long version of the paper.
 \begin{itemize}
    \item The constructed A2A $\T$ for the
        given $\Au$ is such that for every $q \in Q$, there are two copies $q', q'' \in S$ as mentioned earlier. We also introduce new states in $\T$ as required.
    
    \item The A2A includes a sub-A2A that verifies that the input word is a valid parameter word. For every $x_i$, a branch checks that it appears precisely once in the parameter word. 
    
    \item From a run sub-tree whose root is labelled with $(q',i)$ or $(q'',i)$, the A2A verifies that all runs of $\Au$ from $(q,|w(i)|_\square)$ visit $F$ only finitely many times.    
    To do this, for all transitions $\delta$ of the form $q \xrightarrow{op} r$ in $\Au$, we create a sub-A2A $\subata{\delta}$ using copies of sub-A2As. For each such transition, one of two cases should hold: either
    the transition cannot be simulated (because of a zero test or a decrement from zero), or the
    transition can indeed be simulated. For the former, we add a \emph{violation branch} to check that it is indeed the case; for the latter, a
    \emph{validation branch} checks the transition can be simulated and a \emph{simulation branch}
    reaches the next vertex with the updated counter value. Now if the root vertex is of the form $(q',i)$ then the \emph{simulation branch} could reach a vertex labelled with $r'$ or with $r''$ --- with the idea being that $\T$ can choose to move to the safety copy or to stay in the ``normal'' copy of $\Au$. If the root vertex is of the form $(q'',i)$, the simulation branch can only reach the vertex labelled with $r''$ with the updated counter value.
    \item We obtain the global A2A $\T$
        by connecting sub-A2As. To
        ensure that all runs of $\Au$ are 
        simulated, we have the global transition 
        relation $\Delta$ be a conjunction of that of the sub-A2As which
        start at the same state $q \in \{p',p''\}$ for some $p \in Q$. For instance, let
        $\delta_1 = (q,op_1,q_1)$ and $\delta_2=(q,op_2,q_2)$
        be transitions of $\Au$. The constructed sub-A2As $\subata{\delta_1},
        \subata{\delta_2}$
        will contain transitions $(q,\square,\beta_1) \in \Delta_1,
        (q,\square,\beta_2) \in \Delta_2$ respectively.
        In $\T$, we instead have $(q,\square,\beta_1 \land \beta_2) \in \Delta$.
    
    \item Finally, the accepting states are chosen as follows: For every $q \in Q\setminus F$, we set $q''$ as accepting in $\T$. The idea is that if a run in $\Au$ satisfies the coB\"uchi condition then, after some point, it stops visiting target states. In $\T$, the simulated run can choose to move to the safety copy at that point and loop inside it forever thus becoming an accepting branch. On the other hand, if a run does not satisfy the condition, its simulated version cannot stay within the safety copy. (Rather, it will reach the non-accepting sink states.) Also, the violation and the validation branches ensure that the operations along the runs have been simulated properly inside $\T$. It follows that $\T$ accepts precisely those words whose run-tree contains a simulation branch where states from $F$ have been visited only finitely
    many times.
\end{itemize}


\subsection{An upper bound for reachability synthesis of OCAPT}\label{sec_ub_reachsynth}

Following~\cite{bqs19}, we now sketch a guess-and-check procedure using the fact
that \Cref{OCA to A2A} implies a sufficient bound
on valuations satisfying the reachability synthesis problem. Recall that, the reachability synthesis problem asks whether all infinite runs reach a target state.

\begin{lemma}[Adapted from~{\cite[Lemma 3.5]{bqs19}}]\label{bounded-parameter}
    If there is a valuation $V$ such that all infinite $V$-runs of $\Au$ reach
    $F$, there is a valuation $V'$ such that $V'(x)
    =\exp(|\Au|^{\mathcal{O}(1)})$ for all $x \in X$ and all infinite
    $V'$-runs of $\Au$ reach $F$.
\end{lemma}

It remains to give an algorithm to verify that in the resulting non-parametric OCA (after substituting
parameters with their values), all infinite runs from $(q_0,0)$ reach $F$.

\begin{proposition} \label{prop_synthreach_nopara}
    Checking whether all infinite runs from $(q_0,0)$ reach a target state in a  non-parametric one-counter automata is
    \coNP-complete.
\end{proposition}
Before proving the claim above, we first recall a useful lemma from~\cite{lechner15}.

A \emph{path} $\pi= q_0 q_{1} \dots q_{n}$ in $G_\Au$
is a \emph{cycle} if $q_0=q_{n}$. We say the cycle is \emph{simple} if no state (besides $q_0$)
is repeated.
A cycle \emph{starts from a zero test}
if $\delta(q_0,q_1)$ is ``$=0$''. A \emph{zero-test-free cycle} is a cycle where no $\delta(q_i,q_{i+1})$ is a zero test.
We define a \emph{pumpable cycle} as being a zero-test-free cycle such that for all runs 
$\rho= (q_0, c_0) \dots (q_{n},c_{n})$ lifted from $\pi$ we have $c_{n} \geq c_0$, i.e.,
the effect of the cycle is non-negative.

\begin{lemma}[From~\cite{lechner15}] \label{property_of_inf_runs}
    Let $\Au$ be a SOCA with an infinite run that does not reach $F$.
    Then, there is an infinite run of $\Au$ which does not reach $F$ such that it induces a path
    $\pi_0 \cdot \pi_1^{\omega}$, where $\pi_1$ either starts from a zero test or
    it is a simple pumpable cycle.  
\end{lemma}


%

\begin{proof}[Sketch of proof of~\Cref{prop_synthreach_nopara}]
  We want to check whether all infinite runs starting from $(q_0,0)$ reach
  $F$. \Cref{property_of_inf_runs} shows two conditions, one of which must
  hold if there is an infinite run that does not reach $F$. Note that both
  conditions are in fact reachability properties: a path to a cycle that
  starts from a zero test or to a simple pumpable cycle.

  For the first condition, making the reachability-query instances concrete
  requires configuration a $(q,0)$ and a state $q'$ such that $\delta(q,q')$
  is a zero test. Both can be guessed and stored in polynomial time and space.
  For the other condition, we can assume that $\pi_0$ does not have any simple
  pumpable cycle. It follows that every cycle in $\pi_0$ has a zero test or
  has a negative effect. Let $W_{\textrm{max}}$ be the sum of all the positive
  updates in $\Au$. Note that the counter value cannot exceed
  $W_{\mathrm{max}}$ along any run lifted from $\pi_0$ starting from
  $(q_0,0)$.  Further, since $\pi_1$ is a simple cycle the same holds for $2W_{\mathrm{max}}$
  for runs lifted from $\pi_0\pi_1$. Hence, we can guess and store in
  polynomial time and space the two configurations $(q,c)$ and $(q,c')$
  required to make the reachability-query instances concrete.

  Since the reachability problem for non-parametric SOCAP is in
  \NP~\cite{hkow09}, we can guess which condition will hold and guess the
  polynomial-time verifiable certificates.  This implies the problem is in
  \coNP.
	
    For the lower bound, one can easily give a reduction from the complement of the
    \textsc{SubsetSum} problem, which is \NP-complete~\cite{gj79}. The idea is
    similar to reductions used in the literature to prove \NP-hardness for
    reachability in SOCAP. In the long version of the paper, the reduction is given in full detail.
\end{proof}


\section{Conclusion}
We have clarified the decidability status of synthesis problems for
OCA with parameters and shown that, for
several fixed $\omega$-regular properties, they are in \NTWOEXP. If the
parameters only appear on tests, then we further showed that those synthesis
problems are in \PSPACE. Whether our new upper bounds are tight remains an
open problem: neither our \coNP-hardness result for the reachability
synthesis problem nor the \PSPACE{} and $\NP^\NP$ hardness results
known~\cite{sc85,lechner15,lechner16} for other synthesis problems (see
\Cref{tab:complexity}) match them.

We 
believe the BIL fragment will find uses beyond the synthesis
problems for OCA with parameters: e.g. it might imply decidability of
the software-verification problems that motivated the study of \aerpadplus
in~\cite{bi05}, or larger classes of quadratic string equations than the ones
solvable by reduction to \epad~\cite{lm18}. While we have shown BIL is
decidable in \NTWOEXP,
the best
known lower bound is the trivial
\coNP-hardness that follows from encoding the complement of the \textsc{SubsetSum}
problem. (Note that BIL does not syntactically include the $\Pi_1$-fragment of PA
so it does not inherit hardness from the results in~\cite{haase14}.) Additionally, it
would be interesting to reduce validity of BIL sentences
to a synthesis problem. Following~\cite{hkow09}, one can easily
establish a reduction to this effect for sentences of the form:
\(
  \forall \vec{x} \exists \vec{y}
  \bigvee_{i \in I} f_i(\vec{x}) \divs g(\vec{x},\vec{y}) \land
  f_i(\vec{x}) > 0
  \land \varphi_i(\vec{x})
  \land \vec{y} \geq \vec{0}
\)
but full BIL still evades us.


\bibliography{refs}

\clearpage
\onecolumn
\appendix
\section{Lechner's trick}

\begin{proof}[Proof of \Cref{pro:lechner}]
    Consider a sentence $\Phi$ in \aerpad:
    \[
      \forall x_1 \dots \forall x_n \exists y_1 \dots \exists y_m
      \varphi(\vec{x},\vec{y}).
    \]
    We observe $\Phi$ can always be brought into
    negation normal form so that negations are applied only to
    predicates~\cite{rv01}. Hence, it suffices to argue that we can remove
    negated divisibility predicates while staying within \aerpad.
    
    The claim follows from the identity below since the newly
    introduced variables $x',x''$ are both existentially quantified and only
    appear on the right-hand side of divisibility constraints.
    For all $a,b \in \Z$ we have the following. 
    \begin{align*}
        \lnot(a \divs b) \iff (a = 0 \land b \neq 0) \lor {}  \exists x' \exists x''
        \Big( &\left((b = x' + x'') \land (a \divs x') \land (0 < x'' < a)\right) \lor {}\\
        & \left((b = -x' - x'') \land (a \divs x') \land (0 < x'' < -a)\right)\Big)
    \end{align*}
    In other words, if $a = 0$ and $b \neq 0$ then $\lnot(a \divs b)$.
    Further, if $a \neq 0$, there are integers
    $q,r \in \Z$ such that $b = qa + r$ and $0 < r < |a|$ if and only if
    $\lnot (a \divs b)$.
\end{proof}

\section{Undecidability of \texorpdfstring{\earpad}{EARPAD}}
For completeness, we give a proof of Proposition~\ref{pro:undec-earpad} below.
\begin{proof}[Proof of Proposition~\ref{pro:undec-earpad}]
    We will show the language \earpad of all sentences of the form $\lnot
    \varphi$ such that $\varphi \in$ \aerpad is undecidable.

    We begin by recalling the definition of the $\lcm(\cdot,\cdot,\cdot)$ 
    predicate.
    A common multiple of $a,b \in \mathbb{Z}$ is an integer $m \in \mathbb{Z}$
    such that
    $a \divs m$ and $b \divs m$. Their least common multiple $m$ is minimal, 
    that is $m \divs m'$ for all common multiples $m'$. This leads to the 
    following definition of $\lcm(a,b,m)$ for all $a,b,m \in \Z$.
    \[
        \lcm(a,b,m) \iff \forall m'
        \left(( a \divs m') \land (b \divs m'))
        \longleftrightarrow (m \divs m')\right)
    \]
    Observe that the universally-quantified $m'$ appears only on the
    right-hand side of 
    the divisibility constraints. We thus have that \earpad can be assumed to
    include
    a least-common-multiple predicate.\footnote{We remark that this
    definition of the least
    common multiple is oblivious to the sign of $m$, e.g. $\lcm(2,3,-6)$
    is true and
    $\lcm(a,b,m) \iff \lcm(a,b,-m)$ in general.
    This is not a problem since we can add $m \geq 0$ if desired.}
    For convenience, we will write $\lcm(a,b) = m$ instead of
    $\lcm(a,b,m)$.

    Now, once we have defined the $\lcm(\cdot,\cdot,\cdot)$ predicate, we can define the 
    perfect square relation using the identity:
    \[
        x>0 \land   x^2=y \iff \lcm(x,x+1)= y+x
    \]
    and multiplication via:
    \[
        4xy = (x+y)^2 - (x-y)^2.
    \]
    Observe that we are now able to state
    Diophantine equations. Undecidability thus follows from the MRDP theorem~\cite{matijasevic70} which states that satisfiability for such equations (i.e. Hilbert's 10th problem) is undecidable.
\end{proof}

\section{Example where decidability algorithm for \texorpdfstring{\aerpad}{AERPAD} fails}

Here we provide some insight where the attempt of Bozga and Iosif \cite{bi05} fails to show that \aerpadplus is decidable. First note that every \aerpadplus sentence $\varphi$ is of the form where $\Phi= \forall \vec{x} \varphi(\vec{x})$ where $\varphi(\vec{x})=  \exists y_1 \dots \exists y_m \bigvee_{i \in I}
      \bigwedge_{j \in J_i} \left(f_j(\vec{x}) \divs g_j(\vec{x},\vec{y})\right)
      \land \psi_i(\vec{x},\vec{y})$, where $\psi_i$ are Presburger formulas with free variables $\vec{x}$ and $\vec{y}$. In their proposed algorithm the first step claims that by substituting and renaming the existentially quantified variables, we can reduce $\varphi$ to the following DNF-BIL form:
      \begin{equation*}
      \exists y_1 \dots \exists y_m \bigvee_{i \in I}
      \bigwedge_{j \in J_i} \left(f_j(\vec{x}) \divs g_j(\vec{x},\vec{y})\right)
      \land \psi'_i(\vec{x})
      \end{equation*}
      Intuitively their algorithm proposes that we can remove all the existentially quantified variables occurring outside of the divisibility predicates. Now, we take an example: we start with a \aerpadplus formula and follow their proposed steps and show that it is not true. 

\begin{center}
\scalebox{0.9}{
\begin {tikzpicture}[-latex ,auto ,node distance =1 cm]

\node (S) [] {$\begin{aligned} &\exists x_1 \exists x_2  (y \divs 5x_1+4x_2)\\  &\land (5x_1+ 6x_2 -y \leq 0)\\ &\land (5x_1+ 4x_2 -y \leq 0)\\ &\land (3y-2x_2 \leq 0)\end{aligned}$};
\node (A) [right= 6cm of S] {$\begin{aligned} &\exists \vec{x} \exists \vec{z} (y \divs 5x_1+4x_2)\\ &\land (5x_1+ 6x_2 -y +z_1=0)\\ &\land (5x_1+ 4x_2 -y + z_2= 0)\\ &\land (3y-2x_2+z_3=0)\\ &\land (\vec{z} \geq 0)\end{aligned}$};
\node (B) [below= 3cm of A] {$\begin{aligned}&\exists x_2 \exists \vec{z} (y \divs 2y-z_1-2x_2)\\ &\land (y-2x_2-z_1+z_2=0)\\ &\land (3y-2x_2+z_3=0)\\ &\land (\vec{z} \geq 0)\end{aligned}$};
\node (C) [left= 5cm of B] {$\begin{aligned}&\exists \vec{z} (y \divs y-z_2)\\ &\land (2y+z_1-z_2+z_3=0)\\ &\land (\vec{z} \geq 0) \end{aligned}$};

\path (S) edge[] node[sloped,above=0.2cm]{$\text{turning inequalities to equalities}$}(A);
\path (A) edge[] node[left=0.2cm]{$\text{removing }x_1$}(B);
\path (B) edge[] node[sloped,below=0.2cm]{$\text{removing }x_2$}(C);
\end{tikzpicture}
}
\end{center}

  Now the equation $(2y+z_1-z_2+z_3=0)$ cannot be reduced anymore as we cannot remove any of the $\vec{z}$ variables and hence in the end we get existentially quantified variables outside divisibility. \hfill $\blacksquare$
 
\section{Reduction from all the Synthesis Problems to the coB\"uchi one}
\begin{proof}[Proof of \Cref{lem_synth_red}]
Here we give the polynomial time reduction from the reachability, safety and B\"uchi synthesis problems to the coB\"uchi synthesis problem.

Consider a SOCAP $\Au = (Q,T, \delta, X)$, an initial configuration $(q_0, c_0)$ and the set of target states $F$. We construct an automaton $\B = (Q', T', \delta', X)$ which is disjoint union of two copies of $\Au$: $\B \defeq \Au_1 \uplus \Au_2$. We denote the states of $\Au_1$ as $Q_1$ and states of $\Au_2$ as $Q_2$ and the set of target states in $\B$ as $F'$. We take the initial configuration as $(q^{\mathrm{in}}_1, c_0)$ in $\B$ where, $q^{\mathrm{in}}_1 \in Q_1$ is the copy of $q_0$ in $\Au_1$. We ``force'' a move from the first copy to the second one via the target states (only) and there is no way to come back to the first copy once we move to the second one. Formally, for every transition $(u,v) \in T$ such that $u \notin F$, we have $(u_1,v_1), (u_2,v_2) \in T'$ where $u_i,v_i \in Q_i$. For the transitions $(s,t) \in T$ such that $s \in F$, we have $(s_1, t_2), (s_2,t_2) \in T'$ where $s_i,t_i \in Q_i$. For all states $q \in Q_1$ and $q' \in Q_2$, $q \in F'$ and $q' \notin F'$.

Note that, for all valuations, there is an infinite run in $\Au$ that visits a target state if and only if in $\B$ the corresponding run moves to $\Au_2$ (and never comes back to the first copy) if and only if it visits target states only finitely many times. Hence, the answer to the safety synthesis problem in $\Au$ is false if and only if the answer to the B\"uchi synthesis is false in $\B$. For the reduction from \emph{reachability synthesis to B\"uchi}, we can take the exact same construction of $\mathcal{B}$ reversing the target and the non-target states in $\mathcal{B}$. 

The construction of the automaton $\mathcal{B}$ for the reduction from \emph{B\"uchi  synthesis to coB\"uchi} is a bit different from the previous one. Here also, we construct $\B$ as a disjoint union of two copies of $\Au$, but we remove the states in $F$ from the copy $\Au_2$. Also, for every $(u,v) \in T$, we have  $(u_1,v_1), (u_1,v_2), (u_2,v_2) \in T'$. (Note that if $v \in F$ then $(u_1,v_2), (u_2,v_2) \notin T'$ as $v_2$ does not exist.) We set $F' \defeq Q_2$. Now, for all valuations there is an infinite run $\rho$ in $\Au$ that visits $F$ only finitely many times if and only if there is an infinite run in $\B$ that follows $\rho$ within $\Au_1$ until it last visits a state from $F$ and then moves to $\Au_2$ so that it visits states from $F'$ infinitely often. Hence, the answer to the B\"uchi synthesis problem in $\Au$ is negative if and only if it is negative for the coB\"uchi problem in $\B$.
\end{proof}

\section{Putting everything together: encoding reachability into BIL}

\begin{proof}[Proof of \Cref{prop_reach_pbil}]
We first define the formula $\reachnt^{(q,q')}(\vec{x},a,b)$ that is
satisfied by a valuation $V$ of $X \cup \{a,b\}$ iff there is a $V$-run from
$(q, V(a))$ to $(q', V(b))$ without any zero-test transitions. By
\Cref{pro:certificates}, there is such a $V$-run if and only if there is a
$V$-run $\rho$ from $(q, V(a))$ to $(q',
V(b))$ without zero-test transitions and such that:
\begin{itemize}
    \item there exists a configuration $(u,k)$ such that, there is a run
      $\rho_1$ from $(q, V(a))$ to $(u,k)$ that has a type-1 reachability
      certificate;
    \item there exists a configuration $(v,k')$ such that, there is a run
      $\rho_2$ from $(u,k)$ to $(v,k')$ that has a type-3 reachability
      certificate;
    \item there is a run $\rho_3$ from $(v,k')$ to $(q', V(b))$ that has a
      type-2 reachability certificate; and
    \item $\rho = \rho_1\rho_2\rho_3$.
\end{itemize}

We will construct formulas for the sub-automaton obtained by removing from $\Au$ all zero-test transitions.
Now, using \Cref{lem-type1} the
first and the third items above can be encoded as $\exists k
\Psi_1^{(q,u)}(\vec{x},a,k)$ and $ \exists k'\Psi_2^{(v,q')}(\vec{x},k',b)$ such
that the valuation $V$ satisfies them. Also, using \Cref{lem-type1}, the second
item can be encoded as $\exists k \exists k' \Psi_3^{(u,v)}(\vec{x},k,k')$.
Combining all of them, $\reachnt^{(q,q')}(\vec{x},a,b)$ becomes $\exists k
\exists k' \Psi_{\mathrm{reach-nt}}(\vec{x},k,k',a,b)$ where
$\Psi_{\mathrm{reach-nt}}$ looks as
follows.\footnote{Note that in the proof of \Cref{thm-synthreach} we could
  also use this simpler implementation of $\reachnt$. We opted for using one
implemented using $\reach$ to keep the argument self-contained.}
\begin{align*}
\bigvee_{u,v \in Q} \left( 
\Psi_1^{(q,u)}(\vec{x},a,k)
\land \Psi_3^{(u,v)}(\vec{x},k,k')
\land
\Psi_2^{(v,q')}(\vec{x},k',b) \right)
\end{align*}

The formula $\reach^{(q,q')}(\vec{x},a,b)$ expressing general reachability can
then be defined by choosing an ordering on the zero tests.  Formally, let $ZT$
denote the set of all zero-test transitions. We write
$1,\dots,m \in ZT$ to denote an enumeration $(p_1,q_1),\dots,(p_m,q_m)$  of a subset of zero-test
transitions. We define $\reach^{(q,q')}(\vec{x},a,b)$ as:
\[
  \bigvee_{1,\dots,m \in ZT}
  \exists k_0 \dots \exists k_{m+1} \exists k'_0 \dots \exists k'_{m+1}
  \Phi(\vec{x},\vec{k},\vec{k}')
\]
where $\Phi$ is given by:
\[
  \Psi_{\mathrm{reach-nt}}^{(q,p_1)}(\vec{x},k_0,k'_0,a,0)
  \land
  \Psi_{\mathrm{reach-nt}}^{(q_m,q')}(\vec{x},k_{m+1},k'_{m+1},0,b)
  \land \bigwedge_{i=1}^{m-1} \Psi_{\mathrm{reach-nt}}^{(q_i,
  p_{i+1})} (\vec{x},k_i,k'_i,0,0)
\]
In words: for each enumeration of zero-test transitions we take the
conjunction of the intermediate $\reachnt$ formulas as well as $\reachnt$
formulas from the initial configuration and to the final one.

Note that $\reach$ has the required form as every $\Psi$ subformula is in
the required form too. Indeed, the existentially quantified variables in each
$\Psi$ only appear in (the right-hand side of) divisibility constraints and
every divisibility constraint $f(\vec{x}) \divs g(\vec{x},\vec{z})$ appears
conjoined with $f(\vec{x}) > 0$. Also, we have only
introduced an exponential number of disjunctions (over the enumeration of
subsets of zero-test transitions), $2|T| + 4$ new variables (since $m \leq
|T|$) and have not changed the bitsize length of constants after the
construction of the $\Psi$ subformulas. Thus, the bitsize of constants and the
number of variables in $\reach$ remain polynomial and $|\reach|$ is at most
exponential in $|\Au|$.
\end{proof}

  \section{Detailed construction of A2A from OCAPT}
  
  Here we give the detailed constructions of all the sub-A2A for each operation of the OCAPT. The general idea of the construction is given in the \Cref{fig:A2A_subA2A}. In this section, while describing A2A constructions from transitions of the form $(q_i,op,q_j)$, we will represent the $\simu$ branches as $s_i \rightsquigarrow s_j$ for readability, where $s_i$ (similarly, $s_j$) represents $q'_i$ or $q''_i$ corresponding to the normal or the safety copy as described earlier. 

\begin{figure}[ht]
	\centering
	\subfloat[%
	Sub-A2A structure in the normal copy]{
		\scalebox{0.7}{
		\begin {tikzpicture}[-latex ,auto ,node distance =1 cm,
		every state/.style={minimum size=7mm,inner sep=0pt}]
		
		\node[state] (Q) [] {$q'_i$};
		\node (R) [right=0.5cm of Q] {$\lor$};
		\node (P)[right=0.5cm of R]{$\land$};
		\node (vio) [below right= 2cm of R] {};
		\node (T)[right= 1cm of P]{$\lor$};
		\node [state] (U) [above right=1cm of T] {$q'_j$};
		\node [state](V) [below right=1cm of T] {$q''_j$};
		\node (val)[below right= 2cm of P]{};
		\node (U')[right= 0.05cm of U]{(normal copy)};
		\node (V')[right= 0.05cm of V]{(safety copy)};	
		\node (T')[right= 0.2cm of T]{$\simu$};
		
		\path (Q) edge [] node[] {} (R);
		\path (R) edge [] node[] {} (P);
		\path (R) edge [] node[sloped, above] {$\vio$} (vio);
		\path (P) edge [] node[sloped, above] {$\val$} (val);
		\path (P) edge [] node[above] {} (T);
		\path (T) edge [] node[] {} (U);
		\path (T) edge [] node[] {} (V);
	\end{tikzpicture}}
	}
	\qquad
	\subfloat[%
	Sub-A2A structure in the safety copy]{
		\scalebox{0.7}{
		\begin {tikzpicture}[-latex ,auto ,node distance =1 cm,
		every state/.style={minimum size=7mm,inner sep=0pt}]
		
		\node[state] (Q) [] {$q''_i$};
		\node (R) [right=0.5cm of Q] {$\lor$};
		\node (P)[right=0.5cm of R]{$\land$};
		\node (vio) [below right= 2cm of R] {};
		\node[state] (T)[right= 2cm of P]{$q''_j$};
		\node (val)[below right= 2cm of P]{};

		\path (Q) edge [] node[] {} (R);
		\path (R) edge [] node[] {} (P);
		\path (R) edge [] node[sloped, above] {$\vio$} (vio);
		\path (P) edge [] node[sloped, above] {$\val$} (val);
		\path (P) edge [] node[above] {$\simu$} (T);
	\end{tikzpicture}}
	}
	\caption{General Sub-A2A structure simulating $(q_i,op,q_j)$}\label{fig:A2A_subA2A}
\end{figure}

Now we move forward to the detailed constructions for each operations. 
  \begin{figure}[ht]
    \centering
    \subfloat[%
    $\subata{\mathrm{inp}}$ checking if
    the input is a valid parameter
    word\label{fig:A2A_validinput}]{\scalebox{0.7}{
\begin {tikzpicture}[-latex ,auto ,node distance =1 cm,
every state/.style={minimum size=7mm,inner sep=0pt}]

\node (S) [] {$\land$};
\node[state] (S') [left= 1cm of S] {$s_{in}$};
\node[state] (Q) [below=2.5cm of S] {$q'_0$};
\node[state] (A) [right=3cm of S, label = above:$search(x_2)$] {};
\node[below=0.7cm of A] (Ds) {$\vdots$};
\node[state] (B) [right=3cm of A] {$\checkmark_{x_1}$};
\node[state] (C) [right=3cm of Q, label = below:$search(x_n)$] {};
\node[state] (D) [right=2cm of C] {$\checkmark_{x_n}$};

\path (S) edge[dotted] node[sloped,above=0.2cm]{$+1$}(Ds);
\path (S') edge [] node [above =0.2cm]{$\square$} (S);
\path (S) edge [] node {$0$} (Q);
\path (S) edge [] node {$+1$} node[above =0.5 cm] {} (A);
\path (S) edge [] node[sloped, below left =0.1 cm] {$+1$} (C);
\path (A) edge [] node[above =0.15 cm] {$x_1, +1$} (B);
\path (C) edge [] node[below =0.15 cm] {$x_n, +1$} (D);
\path (A) edge [loop below] node[right] {$X\setminus \{x_1\}, +1$} (A);
\path (C) edge [loop above] node[right] {$X\setminus \{x_n\}, +1$} (C);
\path (B) edge [loop below] node[right] {$X\setminus \{x_1\}, +1$} (B);
\path (D) edge [loop above] node[right] {$X\setminus \{x_n\}, +1$} (D);
\end{tikzpicture}
}}
    \qquad
    \subfloat[%
    $\subata{inc}$ encoding an increment
    \label{fig:A2A_increment}]{\scalebox{0.7}{
\begin {tikzpicture}[-latex ,auto ,node distance =1 cm,
every state/.style={minimum size=7mm,inner sep=0pt}]

\node[state] (A) [] {$s_i$};
\node[state] (B) [right=1cm of A,label=below:$\righ(s_j)$] {};
\node[state] (C) [above=2.5cm of B]   {$s_j$};

\path (A) edge [] node {$\square, +1$} (B);
\path (B) edge [] node[swap] {$\square, 0$} (C);
\path (B) edge [loop right] node[above =0.15cm] {$x, +1$} (B);
\end{tikzpicture}
}}
    \caption{Sub-A2As for the word-validity check and to simulate
    increments of the form $(q_i,+1,q_j)$; we use $search(x)$, $\checkmark_{x}$, and $\righ(q)$ as state names to make their function explicit}
\end{figure}
  
\subparagraph{Verifying the input word} 
The sub-A2A $\subata{inp}$ depicted in \Cref{fig:A2A_validinput}
checks whether the given input is a valid parameter word. The states of the form $\checkmark_{x_i}$ represents that $x_i$ has been found along the path.
We let $S_f$ consist of states $\checkmark_{x_i}$, one per $x_i \in X$. 

\begin{lemma}\label{lem_A2A_validinput}
     It holds that $L(\subata{inp})= W_X$.
\end{lemma}

\begin{proof}
    The A2A $\subata{inp}$ consists of one deterministic one-way automata, per $x \in X$,
    whose language clearly corresponds to 
    the set of words where $x$ occurs exactly once.
    In $\subata{inp}$, from the initial state
    and on the first letter $\square$, a
    transition with a conjunction formula
    leads to all sub-automata for each $x$.
    The result follows.
\end{proof}

\subparagraph{Increments}
For every transition $\delta= (q_i,+1,q_j)$ we construct $\subata{inc}$
(see \Cref{fig:A2A_increment}). A run of this sub-A2A
starts from $s_i$ and some position $c$ on the input word. Recall that $c$
uniquely determines the current counter value in the simulated run of
$\Au$ (although, it should be noted $c$ itself is not the counter value). Then, the run of $\subata{inc}$ moves to the next occurrence of
$\square$ to the right of the current position and then goes to $s_j$ accordingly.

\subparagraph{Decrements}
For transitions $\delta = (q_i,-1, q_j)$ we construct $\subata{dec}$ 
(see \Cref{fig:A2A_decrement}). In contrast to the increment sub-A2A,
it also includes a $\vio$ branch in case the decrement would result
in a negative counter value: On this branch, $\subata{dec}$ attempts to
read $\first?$ to determine if the position of the head corresponds to the
first letter of the word.

\begin{figure}[ht]
    \centering
    \subfloat[%
    $\subata{dec}$ encoding an
    decrement\label{fig:A2A_decrement}]{\scalebox{0.7}{
\begin {tikzpicture}[-latex ,auto ,node distance =1 cm,
every state/.style={minimum size=7mm,inner sep=0pt}]

\node (A) [] {$\lor$};
\node[state] (Q) [left =1cm of A] {$s_i$};
\node[state] (B) [above right=1cm of A,label = above:$\lef(s_j)$] {};
\node[state] (C) [right=2cm of B]   {$s_j$};
\node[state] (D) [right=2cm of A,label = below:$\final$] {};
\node[state] (E) [right=2cm of D,label = below:$\true$] {};

\path (Q) edge [] node[above =0.2 cm] {$\square$} (A);
\path (A) edge [] node[pos=0.6,sloped,below =0.15 cm] {$-1$} (B);
\path (B) edge [] node[above =0.15 cm] {$\square,0$} (C);
\path (B) edge [loop left] node[left =0.15cm] {$x, -1$} (B);
\path (A) edge [] node[below =0.15cm] {$0$} (D);
\path (D) edge [] node[below =0.15cm] {$\first?$} (E);
\end{tikzpicture}
}}
    \qquad
    \subfloat[%
    $\subata{zero}$ encoding a zero
    test\label{fig:A2A_zerocheck}]{\scalebox{0.7}{
\begin {tikzpicture}[-latex ,auto ,node distance =1 cm,
every state/.style={minimum size=7mm,inner sep=0pt}]

\node (A) [] {$\lor$};
\node[state] (Q) [left =1cm of A] {$s_i$};
\node[state] (B) [above right=1cm of A] {$s_i^{=0}$};
\node[state] (C) [right=2cm of B]   {$s_j$};
\node[state] (D) [right=2cm of A] {$s_i^{\neq 0}$};
\node[state] (E) [right=2cm of D,label=above:$\true$] {};

\path (Q) edge [] node[above =0.15 cm] {$\square$} (A);
\path (A) edge [] node[sloped, above =0.15 cm] {$0$} (B);
\path (B) edge [] node[above =0.15 cm] {$\first?$} (C);
\path (A) edge [] node[sloped, below =0.15cm] {$-1$} (D);
\path (D) edge [] node[below =0.15cm] {$\Sigma$} (E);
\end{tikzpicture}
}}
    \caption{Sub-A2As to simulate decrements and zero tests}
\end{figure}
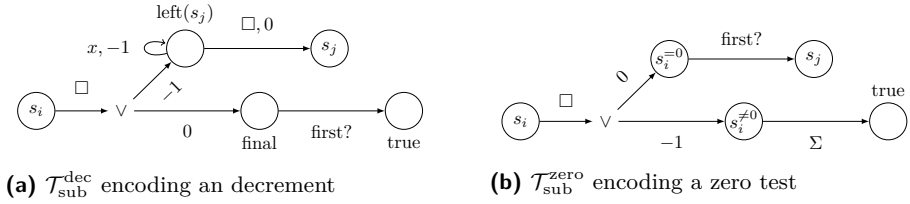

\begin{lemma} \label{lem_A2A_decrement}
    Let $k,l \in \N$ and $w \in W_X$ with $\square$ the $(i+1)$-th letter
    of $w$.
    A run tree $\gamma$ of $\subata{dec}$ on $w$ from $k$
    is accepting if and only if
    either $(s_i,k) \rightsquigarrow (s_j,l)$
    is a part of $\gamma$ and $|w(k)|_\square - 1 = |w(l)|_\square$,
    or $(s_i,0) \rightsquigarrow (\mathrm{final}, 0)$ is a part of $\gamma$
    and $k = 0$.
\end{lemma}

 \begin{proof}
    Note that any accepting run $\gamma$ of the sub-A2A must include
    at least one of the two finite branches from the claim.
    We further argue that each branch enforces
    the corresponding constraints if they appear in $\gamma$.
    Since these are
    mutually exclusive, it follows that $\gamma$
    includes exactly one of the branches.
    
    If $\gamma$ includes $(s_i,k) \rightsquigarrow (s_j,l)$
    then $|w(k)|_\square - 1 = |w(l)|_\square$. The latter implies
    $k > l >= 0$ since otherwise the position of the head
    cannot be moved to the left. On
    the other hand, if $\gamma$ includes $(s_i,n) \rightsquigarrow
    (\mathrm{final},n)$ then $\gamma$ can only be accepting if $n = 0$.
    Hence, $\gamma$ includes $(s_i,0) \rightsquigarrow (\mathrm{final},0)$.
\end{proof}

\subparagraph{Zero tests}
For every transition $\delta = (q_i, =0, q_j)$ we construct $\subata{zero}$
(see \Cref{fig:A2A_zerocheck}) similarly to how we did for decrements.
For the $\val$ branch, it reads $\first?$ to confirm the position of
the head is at the beginning of the word. For the $\vio$ branch, it moves
the head to the left to confirm that the head is not at the beginning.

\begin{lemma} \label{lem_A2A_zerocheck}
    Let $k \in \N$ and $w \in W_X$ with $\square$ the $(k+1)$-th letter
    of $w$. A run tree $\gamma$ of $\subata{zero}$ on $w$ from $k$
    is accepting if and only if
    either $(s_i,0) \rightsquigarrow (s_j,0)$ is a part of $\gamma$ and $k=0$,
    or $(s_i,k) \rightsquigarrow (s_i^{\neq 0}, k-1)$ is a part of $\gamma$
    and $|w(k)|_\square > 0$.
\end{lemma}

\begin{proof}
    We proceed as in the proof of \Cref{lem_A2A_decrement}.
    
    If $\gamma$ includes a branch with the state
    $s_i^{=0}$ then $\gamma$ is accepting if and only if it reaches $s_j$.
    It can only reach $s_j$ with the $\first?$ transition, i.e. when $k=0$. 
    Otherwise, it has to include a branch with $s_i^{\neq 0}$ and reading
    any letter it reaches $\true$. This is only possible if $k > 0$.
    Since the $(k+1)$-th letter of $w$ is $\square$, the
    latter means $|w(k)|_\square > 0$.
\end{proof}

\subparagraph{Parametric equality tests}
For every transition $\delta = (q_i, =k, q_j)$ we construct
$\subata{eq}$
(see \Cref{fig:A2A_equality}). For the $\val$ branch, it moves the 
head right, skipping over other variable symbols $X \setminus\{x\}$, 
while looking for $k$. For the $\vio$ branch it skips over other
variable symbols while looking for the next $\square$.

\begin{figure}[ht]
    \centering
    \subfloat[%
    $\subata{eq}$ for parametric equality
    tests\label{fig:A2A_equality}]{\scalebox{0.7}{
\begin {tikzpicture}[-latex ,auto ,node distance =1 cm,
every state/.style={minimum size=7mm,inner sep=0pt}]

\node[state] (Q) {$s_i$};
\node (lor) [below=0.8cm of Q] {$\lor$};
\node (land) [right=1.5cm of lor] {$\land$};
\node[state] (Q') [above=0.8cm of land] {$s_j$};
\node[state] (A) [right=1.5cm of land, label=below left:$present(x)$] {};
\node[state] (B) [below=1.5cm of A, label=below:$\true$] {};
\node[state] (C) [left=1.5cm of B, label=below:$absent(x)$] {};

\path (Q) edge [] node[swap] {$\square$} (lor);
\path (lor) edge  (land);
\path (land) edge [] node {$0$} (Q');
\path (land) edge [] node[above =0.15 cm] {$+1$} (A);
\path (lor) edge [] node[sloped] {$+1$} (C);
\path (A) edge [] node {$x$}  (B);
\path (C) edge [] node {$\square$}   (B);
\path (A) edge [loop right] node {$X\setminus \{x\}, +1$} (A);
\path (C) edge [loop left] node[] {$X\setminus \{x\}, +1$} (C);
\end{tikzpicture}
}}
    \qquad
    \subfloat[%
    $\subata{lb}$ for parametric lower-bound
    tests\label{fig:A2A_lowerbound}]{\scalebox{0.7}{
\begin {tikzpicture}[-latex ,auto ,node distance =1 cm,
every state/.style={minimum size=7mm,inner sep=0pt}]

\node[state] (Q) {$s_i$};
\node (lor) [below=1cm of Q] {$\lor$};
\node (landr) [right= 1cm of lor] {$\land$};
\node[state](less)[below right=1.5cm of lor, label=left:$front(x)$] {};
\node (lorr) [right= 1cm of landr] {$\lor$};
\node (equal)[right= 1cm of lorr,align=left]{$\land$-state\\ in $\subata{eq}$};
\node(search)[right=1.5cm of less,align=left]{$search(x)$\\ in $\subata{inp}$};
\node[state] (Q') [above=1cm of landr] {$s_j$};
\node[state] (A) [above right=1.2cm of lorr,label=left:$back(x)$] {};
\node[state] (B) [right=1.2cm of A,label=above:$search^{-}(x)$] {};
\node[state] (C) [below=2cm of B,label=below:$\true$] {};

\path (Q) edge [] node[swap] {$\square$} (lor);
\path (lor) edge [] (landr);
\path (landr) edge (lorr);	
\path (landr) edge node {$0$} (Q');
\path (lorr) edge node[swap,sloped] {$-1$}  (A);
\path (lorr) edge node[swap] {$0$} (equal);
\path (A) edge [] node {$\square, -1$}   (B);
\path (B) edge [] node {$x$}  (C);
\path (B) edge [loop right] node {$\Sigma\setminus \{x\}, -1$}  (B);
\path (lor) edge [] node[sloped,below =0.15 cm] {$+1$} (less);
\path (less) edge[] node {$\square,+1$} (search);
\path (less) edge [loop above] node {$X\setminus \{x\}, +1$} (less);
\path (A) edge [loop above] node {$X\setminus \{x\}, -1$} (A);
\end{tikzpicture}
}}
    \caption{Sub-A2As to simulate parametric tests}
\end{figure}
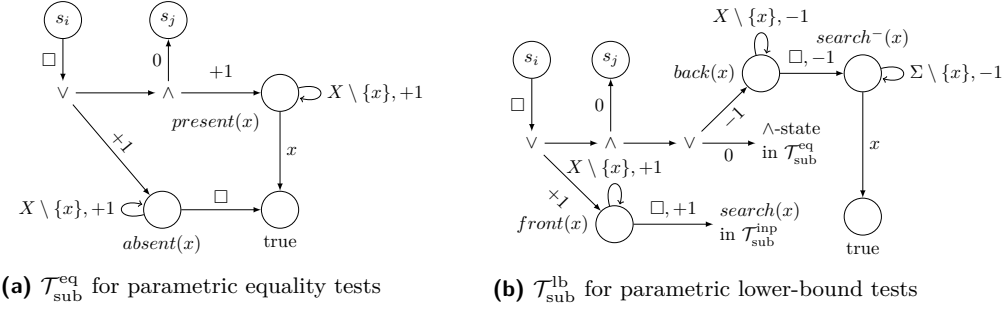

\begin{lemma} \label{lem_A2A_equality}
    Let $k \in \N$ and $w \in W_X$ with $\square$ the $(k+1)$-th
    letter of $w$. A run tree $\gamma$ of $\subata{eq}$ on $w$
    from $k$ is accepting if and only if
    either $(s_i,k) \rightsquigarrow (s_j,k)$ is part of $\gamma$ and
    $V_w(x) = |w(k)|_\square$,
    or $(s_i,k) \rightsquigarrow (absent(x), k+1)$ is a part of
    $\gamma$ and $V_w(x) \neq |w(k)|_\square$.
\end{lemma}

\begin{proof}
    Fix a word $w \in W_X$ with $\square$ as $(k+1)$-th letter. Consider any run tree $\gamma$ of $\subata{eq}$ on $w$.
    After reading the first $\square$, suppose $\gamma$ has a branch
    leading to the state $s_j$. It must therefore also have a branch containing
    $present(x)$. Since, from there, it can only move to the state $\true$ if it reads $x$
    before reading another $\square$ symbol to the right, we have $V(x)=|w(k)|_\square$.

    If $\gamma$ has a branch containing $absent(x_i)$,
    then it is accepting if and only if it reaches
    $\true$ after reading another $\square$
    before ever reading $x$. Hence, $V(x_i) \neq |w(k)|_\square$.
\end{proof}

\subparagraph{Parametric lower-bound tests}
For every transition $\delta = (q_i, \geq x, q_j)$ we
construct $\subata{lb}$ (see \Cref{fig:A2A_lowerbound}).
For the $\val$ branch, we check for equality to $x$
or we check whether $>x$.
We also create the corresponding $\vio$ branches.

\begin{lemma} \label{lem_A2A_lowerbound}
    Let $k \in \N$ and $w \in W_X$ with $\square$ the
    $(k+1)$-th letter of $w$. A run tree $\gamma$ of
    $\subata{lb}$ on $w$ from $k$ is accepting if
    and only if
    either $(s_i,k) \rightsquigarrow (s_j, k)$ is part of $\gamma$ and $|w(k)|_\square \geq V_w(x)$,
    or $(s_i,k) \rightsquigarrow (front(x), k+1)$ is a part of $\gamma$ and $|w(k)|_\square < V_w(x)$.
\end{lemma}

\begin{proof}[Proof of \Cref{lem_A2A_lowerbound}]
    Fix a word $w \in W_X$ with $\square$ as $(k+1)$-th letter and consider any run tree $\gamma$ of $\subata{lb}$ on $w$. After reading the first $\square$, let us
    suppose it adds a branch checking $= x$
    in $\subata{eq}$. Then, $\gamma$ is accepting if and only if it additionally contains a branch to $(s_j, k)$ and $|w(k)|_\square=V_w(x)$. If it has the other sub-tree, i.e. it contains $back(x)$, $\gamma$ is accepting if and only if it reaches the state $\true$ which is possible only if there is a $\square$ to the left of the current position and it reads an $x$ to the left of that. It follows that it is accepting
    if and only if $|w(k)|_\square > V_w(x)$ and $(s_i,k) \rightsquigarrow (s_j, k)$ is part of $\gamma$.

    If $\gamma$ instead contains the branch with $front(x_i)$, it is accepting only if it can read $x$ from $search(x)$
    after having read a $\square$ from $front(x)$ to the right
    of the current position of the input. Hence,
    $|w(k)|_\square < V_w(x_i)$.
\end{proof}

Using the previous lemmas, it is straightforward to prove \Cref{OCA to A2A}. 
The detailed proof of correctness is given below.
\begin{proof}[Proof of~\Cref{OCA to A2A}]
	
	Here we present a detailed proof of \Cref{OCA to A2A}. We have to show that, $L(\T) = \{w \in W_X \mid$ all infinite $V_w$-runs of $\Au$ visit $F$ finitely many times from $(q_0,0)\}$. We prove this in two parts:

    $\supseteq$: Consider a word $w = a_0 a_1 a_2 \dots \in W_X$, such that with valuation
    $V_w$ all infinite $V_w$-runs of $\Au$ visit $F$ only finitely many times starting from $(q_0,0)$. We have to show, that $w$ is accepted
    by
    $\T$, i.e., there exists an accepting run tree $\gamma$ of $w$ on $\T$. 
     We will now grow an accepting run tree $\gamma_{\mathrm{valid}}$.
    Since $w$
    is a valid parameter word, we can add to
    $\gamma_{\mathrm{valid}}$ a sub-tree with root labelled by $(s_{in},0)$
    and a branch extending to $(q'_0,0)$ (see Lemma \ref{lem_A2A_validinput}).
    
    Consider now a valid infinite run $\rho$ of $\Au$ that visits $F$ only finitely many times. Hence, $\rho$ can be divided into $\rho= \rho_f \cdot \rho_{inf}$ such that $\rho_f$ is a finite prefix and $\rho_{inf}$ is the infinite suffix that never visits $F$. Let $\pi$ be the path of the form $(q_0, op_1, q_1)(q_1, op_2, q_2 ) \dots$ induced by $\rho$. We extend the division of $\pi$ into $\pi= \pi_1 \cdot (q_{j-1},op_j,q_{j}) \cdot \pi_2$ such that, $\pi_1 \cdot (q_{j-1},op_j,q_{j})$ is induced by $\rho_f$ and $\pi_2$ is induced by $\rho_{inf}$. The idea is that, the run $\rho$ jumps to a ``safety component'' from the state $q_j$ after which it does not visit $F$ at all as $\rho$ satisfies the coB\"uchi condition.
    
    Now, we further extend $\gamma_{\mathrm{valid}}$ by appending to it, from the $(q'_0,0)$-labelled vertex, a sub-tree $\gamma_{\pi_1}$ simulating the prefix $\pi_1$ as follows: for every transition of the
    form $(q_i, op_{i+1}, q_{i+1})$ where $op_i$ is an increment or decrement, the corresponding $\subata{inc}$ and $\subata{dec}$ simulate the path from $q'_i$ to
    $q'_{i+1}$ correctly. Also, as every transition in $\pi$ is valid in $\pi_1$ (i.e. does not result in negative counter values), using the first part of Lemmas \ref{lem_A2A_zerocheck}, \ref{lem_A2A_equality}, and
    \ref{lem_A2A_lowerbound}, we can take the $\val$ sub-trees of $\subata{zero}$, $\subata{eq}$, and $\subata{lb}$, and append them to our run tree. For every $\simu$ branch, we stay at the normal copy and we move from $q'_i$ to
    $q'_{i+1}$. Now, for the transition $(q_{j-1},op_j,q_{j})$, we do the same for the $\vio$ and $\val$ branches but in the $\simu$ branch, we move to the \emph{safety copy} and move to $q''_j$. Intuitively, this safety copy simulates the safety component of $\rho$ as mentioned above. Now, with this we append another sub-tree $\gamma_{\pi_2}$, which we create exactly in the similar way as $\gamma_{\pi_1}$ but the $\simu$ branch stays in the safety copy, i.e., it moves from states of the form $q''_i$ to $q''_{i+1}$. It is easy to see that, $\gamma_{\pi_2}$ simulates the suffix $\rho_{inf}$ correctly.
    Note that, since $\pi_2$ does not visit $F$ at all, the $\simu$ branch never reaches the non-accepting sink states in the safety copy and it infinitely loops within the accepting states in the safety copy, making it accepting.

	As $\rho$ was chosen arbitrarily, we have that $\gamma_{\rho}$, for all infinite runs $\rho$, are accepting. To conclude, we need to deal with run trees arising from maximal finite runs--the runs that cannot be continued with any valid operation and hence, finite: We construct a sub-tree $\gamma_{\mathrm{maxf}}$ appending $\simu$ and $\val$ sub-trees for as long as possible. By definition of maximal finite runs, every such run reaches a point where all possible transitions are disabled. There, we append a $\vio$ sub-tree which, using the second part of the mentioned lemmas, is accepting. Hence, $\gamma_{\mathrm{valid}}$ is accepting.
    
    $\subseteq:$ Consider a word $w \in L(\T)$. We have to show that with valuation $V_w$, every infinite run of $\Au$ visits $F$ only finitely often from $(q_0,0)$. We will prove the contrapositive of this statement: Let there exists a valuation $V$ such that there is an infinite run of $\Au$ that visits $F$ infinitely often from $(q_0,0)$, then for all words $w$ with $V_w=V$, $w \not\in L(\T)$.
    
    Let $\rho$ be such an infinite run with valuation $V_w$. Now, $\rho$ induces the path $\pi$ which has the following form $(q_0, op_1, q_1) \dots$, where for every $i$ there exists a $j$ such that $q_j \in F$. Recall that for every $op_i$, a run of $\subata{op_i}$ has one $\simu$ branch, one or more $\val$ branches or a $\vio$ branch. Now, as $\rho$ is a valid infinite run of $\Au$, every $op_i$ can be taken, i.e, the counter value never becomes negative along the run. Hence, any $\vio$ branch in any $\subata{op_i}$ will be non-accepting already using the corresponding lemmas of the different operations. 
    Hence, for every $op_i$ appearing in $\pi$, let us consider the $\simu$ and $\val$ branches. Consider the global $\simu$ branch $b$ in $\T$: $(s_{in},0) \rightsquigarrow s_0  \rightsquigarrow s_1 \dots$, where each $s_i$ in $\T$ represents $q_i$ in $\Au$ and is in the form $q'_i$ or $q''_i$ depending on whether it has jumped to the safety copy or not. If every $s_i$ is of the form $q'_i$, then the infinite branch $b$ has never moved to the safety copy and has not visited the accepting states at all. Hence, it is already non-accepting. 
    
    Now, for some $l$, let $s_l$ be of the form $q''_l$ representing $q_l$ in $\Au$, i.e., it has moved to the safety copy in $\T$. Note that, if a branch in $\T$ moves to a safety copy, it can never escape that is for all $m \geq l$, $s_m$ is of the form $q''_m$. Notice that from our assumption, there exists $n \geq l$, such that $q_n \in F$. Hence, $s_n$, representing $q_n$ in the safety copy of $\T$, is a non-accepting sink establishing the fact that the branch $b$ reaches a non-accepting sink making it non-accepting.

     Note that, $b$ is a valid infinite branch in a run in A2A with no final states visited. Branch $b$ will be present in every run of $w$ in $\T$, resulting no accepting run for $w$.
 \end{proof}

\section{Missing Proofs from Section \ref{sec_ub_reachsynth}}

\begin{proof}[Proof of~\Cref{bounded-parameter}]
    Using \Cref{OCA to A2A} for OCAPT $\Au$, there is an A2A $\T$ of polynomial size (w.r.t. $\Au$)
    such that, $L(\T)$ is precisely the subset of $W_X$
    such that all infinite $V_w$-runs of $\Au$ reach $F$.
    We then use that
    there is a non-deterministic B\"uchi
    automaton $\B$ such that $L(\B) = L(\T)$ and $|\B| \in
    2^{\mathcal{O}(|\T|^2)}$~\cite{vardi98,dk08}.

    Suppose $\Au$ is a positive instance of the reachability synthesis
    problem, i.e. $L(\B) \not = \emptyset$. We know that the language of a
    B\"uchi automata is non-empty only if there is a ``lasso'' word which
    witnesses this. For all parameter words $w$ accepted by a lasso there is a
    word $u \in \Sigma^{*}$ s.t. $|u| \leq |\B|$ and $w = u
    {\square}^{\omega} \in L(\B)$. The result follows from our encoding
    of valuations. 
\end{proof}

\begin{proof}[Proof of \Cref{property_of_inf_runs}]
    Let us call an infinite run of $\Au$ a \emph{safe run} if it does not reach $F$. Fix a safe run $\rho$. Let $\pi = (q_0, op_1, q_1) (q_1, op_2,q_2) \dots$ be the path it induces. We denote
    by $\pi[i,j]$ the infix $(q_i, op_{i+1},q_{i+1}) \dots (q_{j-1}, op_{j}, q_{j})$ of $\pi$
    and by $\pi[i,\cdot]$ its infinite suffix $(q_i, op_{i+1},q_{i+1}) \dots$
    Suppose there are $0 \leq m<n \in \N$ such that $\pi[m,n]$ is a cycle
    that starts from a zero test. Note that if a cycle that starts from a zero test can be traversed
    once, it can be traversed infinitely many times. Then, the run lifted from the path $\pi[0,m]
    \cdot {\pi[m,n]}^{\omega}$ is our desired safe run. Now, let us assume that $\pi$ has no cycles which
    start at a zero test. This means every zero test occurs at most once in $\pi$. Since the
    number of zero tests in $\Au$ is finite, we have a finite $k \in \N$ such that there
    are no zero tests at all in $\pi[k,\cdot]$.

    Now, consider $\pi[k, \cdot]$. Suppose it does not
    witness any non-negative effect cycle, i.e., every
    cycle in $\pi[k, \cdot]$ is negative. But, we know
    $\pi$ lifts to a valid infinite run which means the
    counter value cannot go below zero. This contradicts
    our assumption; Hence, there are $k \leq p < q$
    such that $\pi[p,q]$ is a cycle with non-negative
    effect.
    It is easy to see that there must be $r,s$ such that
    $p \leq r < s \leq q$ and $\pi[r,s]$ is a simple
    non-negative effect cycle. Also note that, $r \geq k$
    which means that $\pi[r,s]$ does not have any zero
    tests. Hence, $\pi[r,s]$ is a simple pumpable cycle.
    Note that if a pumpable cycle can be
    traversed once then it can be traversed infinitely
    many times. Using this fact, the run lifted from
    $\pi[0,r] \cdot {\pi[r,s]}^{\omega}$ is our
    desired safe run.
\end{proof}

\begin{figure}[t]
    	\centering
    	\scalebox{0.8}{
\begin {tikzpicture}[-latex ,auto ,node distance =1 cm,initial text={},inner sep=0em, minimum size=1mm,
every state/.style={minimum size=7mm,inner sep=0pt}]

\node[state,initial] (start) [] {$q_0$};
\node[state] (a) [right=1.5cm of start] {$q_1$};
\node[state] (b)[above right= 0.75cm of start]{$q'_1$};
\node[state] (c) [right=1cm of a] {$q_{n-1}$};
\node[state] (d)[above right= 0.75cm of c]{$q'_n$};
\node[state] (e) [right=1.5cm of c] {$q_n$};
\node[state] (f) [right=1cm of e] {$q_t$};
\node[state] (g) [above right=0.75cm of f] {$q_s$};
\node[state, accepting] (h) [right=1cm of f] {$q_f$};

\path (start) edge [] node[above=0.2cm] {$+0$} (a);
\path (start) edge [] node[sloped,above=0.2cm] {$+a_1$} (b);
\path (b) edge [] node[sloped,above=0.2cm] {$+0$} (a);
\path (c) edge [] node[above=0.2cm] {$+0$} (e);
\path (c) edge [] node[sloped,above=0.2cm] {$+a_n$} (d);
\path (d) edge [] node[sloped,above=0.2cm] {$+0$} (e);
\path (a) edge [dotted] node[] {} (c);
\path (e) edge [] node[above=0.2cm] {$-t$} (f);
\path (f) edge [] node[sloped, above=0.2cm] {$=0$} (g);
\path (f) edge [] node[above=0.2cm] {$+1$} (h);
\path (g) edge [loop right] node[right=0.2cm] {$+1$} (g);
\path (h) edge [loop right] node[right=0.2cm] {$+1$} (h);

\end{tikzpicture}
}
  \caption{Reduction from non-\textsc{SubsetSum} to (universal) reachability for SOCA} \label{fig:subsetsum_soca}
    \end{figure}
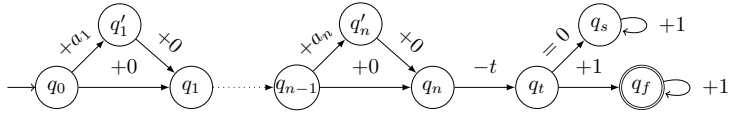

\begin{proof}[Proof of second part of ~\Cref{prop_synthreach_nopara}]
Here we give the full reduction from the complement of the \textsc{SubsetSum} problem to the problem of checking if all infinite runs reach a target state in a non-parametric SOCAP. 

Given a set $S=\{a_1,a_2, \dots a_n\} \subseteq \N$ and a target sum $t \in \N$, the \textsc{Subset Sum} problem asks whether
	there exists $S' \subseteq S$ such that $\sum_{a_i \in S'} a_i =t$. Given an instance
	of the \textsc{Subset Sum} problem with $S$ and $t$, we create a  SOCA $\Au$ with initial configuration $(q_0,0)$ and a single target state $q_f$ as depicted in	\Cref{fig:subsetsum_soca}. Note that, for every $1 \leq i \leq n$ there are two ways of reaching $q_i$ from $q_{i-1}$: directly, with constant update $+0$; or via $q'_i$ with total effect $+a_i$. Hence, for every subset $S' \subseteq S$, there exists a path from $q_0$ to $q_{n}$ with counter value $\sum_{a_i \in S'} a_i$. Clearly, if there exists $S'$ such that $\sum_{a_i \in S'} a_i = t$ \textsc{Subset Sum} then there exists an infinite run leading to $q_s$--not reaching the target state. On the other hand, if there is no such $S'$ then all infinite runs reach $q_f$. Hence, the universal reachability in $\Au$ is positive if and only if the answer to the \textsc{SubsetSum} problem is negative.
\end{proof}

\end{document}